\documentclass[a4paper,11pts]{article}
\usepackage{graphicx,amsmath,amssymb}
\usepackage{subcaption}
\usepackage{hyperref}
\usepackage{paralist}
\usepackage{color}
\usepackage{tikz,pgfkeys}
\usepackage{tkz-graph}
\usetikzlibrary{arrows,shapes}
\usepackage{charter,eulervm,bbding}
\usepackage{setspace}
\usepackage{tagging}
\usepackage{multirow}

\usepackage{amsthm}
\newtheorem{theorem}{Theorem}[section]
\newtheorem{lemma}[theorem]{Lemma}
\newtheorem{corollary}[theorem]{Corollary}

\newtheorem{claim}{Claim}

\def\boxit#1{\vbox{\hrule\hbox{\vrule\kern4pt
  \vbox{\kern1pt#1\kern1pt}
\kern2pt\vrule}\hrule}}

\usetag{full}

\begin{document}

\title{ \bf An Efficient Branching Algorithm for Interval Completion}

\iftagged{short}
{
\author{Yixin Cao}
\institute{Institute for Computer Science and Control\\
  Hungarian Academy of Sciences (MTA SZTAKI)\\
  {\tt yixin@sztaki.hu}} 
}
{
\author{
  {\sc Yixin Cao\thanks{Institute for Computer Science and Control,
      Hungarian Academy of Sciences.  Email: {\tt yixin@sztaki.hu}.
      Supported by the European Research Council (ERC) grant
      ``PARAMTIGHT: Parameterized complexity and the search for tight
      complexity results,'' reference 280152.}  
  }}
  \date{}
}

\maketitle

\begin{abstract}
  We study the \emph{{interval completion}} problem, which asks for
  the insertion of a set of at most $k$ edges to make a graph of $n$
  vertices into an interval graph.  We focus on chordal graphs with no
  small obstructions, where every remaining obstruction is known to
  have a shallow property.  From such a shallow obstruction we single
  out a subset $6$ or $7$ vertices, called the frame, and $5$ missed
  edges in the subgraph induced by the frame.  We show that if none of
  these edges is inserted, then the frame cannot be altered at all,
  and the whole obstruction is also fixed, by and large, in the sense
  that their related positions in an interval representation of the
  objective interval graph have a specific pattern.  We propose a
  simple bounded search process, which effectively transforms a given
  graph to a graph with the structural property that all obstructions
  are shallow and have fixed frames.  Then we fill in polynomial time
  all obstructions that have been previously left in indecision.
  These efforts together deliver a simple parameterized algorithm of
  time $6^k\cdot n^{O(1)}$ for the problem, significantly improving
  the only known parameterized algorithm of time $k^{2k}\cdot
  n^{O(1)}$.
\end{abstract}

\section{Introduction} 
A graph is an \emph{interval graph} if its vertices can be assigned to
the intervals of the real line such that there is an edge between two
vertices if and only if their corresponding intervals intersect.
Interval graphs are the natural models for DNA chains in biology
\cite{benzer-59-topology-genetic-structure} and many other
applications, among which the most cited ones include jobs scheduling
in industrial engineering \cite{bar-noy-01-resource-allocation} and
seriation in archeology \cite{kendall-69-seriation-1}.  Motivated by
pure contemplation of combinatorics and practical problems of biology
respectively, Haj{\'o}s \cite{hajos-57-interval-graphs} and Benzer
\cite{benzer-59-topology-genetic-structure} independently initiated
the study of interval graphs.  An interval graph $\widehat G$ is
called an \emph{interval supergraph} of $G$ if they have the same
vertex set and every edge of $G$ also appears in $\widehat G$.  The
minimum \textsc{interval completion} problem asks for the minimum size
of interval supergraphs of a given graph; or equivalently, the minimum
number of edges whose insertion transforms a graph into an interval
graph.  Originally formulated in sparse matrix computations
\cite{rose-72-sparse-matrix}, this problem later found application in
physical mapping of DNA \cite{kaplan-99-chordal-completion}.  A
similar and related problem is \textsc{chordal completion}, which is
also widely known as \textsc{minimum fill-in}.  A graph is
\emph{chordal} if it contains no hole, and the \textsc{chordal
  completion} problem asks for the minimum size of chordal supergraphs
of a given graph.

These problems are, understandably, NP-hard \cite{GJ79,
  yannakakis-81-minimum-fill-in}.  Therefore, early work of Kaplan et
al.~\cite{kaplan-99-chordal-completion} and
Cai~\cite{cai-96-hereditary-graph-modification} focused on their
parameterized tractability.  Recall that a problem, parameterized by
$k$, is {\em fixed-parameter tractable (FPT)} if it admits an
algorithm with runtime $f(k)\cdot n^{O(1)}$, where $f$ is a computable
function depending only on $k$ \cite{downey-fellows-99}.
Cai~\cite{cai-96-hereditary-graph-modification} observed that if a
hereditary graph class ${\cal G}$ can be characterized by a finite
number of forbidden (induced) subgraphs, then the fixed-parameter
tractability of the completion problem follows from a basic bounded
search tree algorithm.  Many important graph classes, however, have
minimal obstructions of arbitrary large size; in particular, holes of
any length are forbidden in both interval and chordal graphs.  Even
so, {\sc chordal completion} can still be solved by a bounded search
tree algorithm by observing that a large hole readily implies a
negative answer to the problem \cite{kaplan-99-chordal-completion}.

An interval graph is known to be chordal and not contain a structure
called ``asteroidal triple'' (or AT for short), i.e., three vertices
of which each pair is connected by a path avoiding neighbors of the
third one \cite{lekkerkerker-62-interval-graphs}.  Therefore, to solve
\textsc{interval completion}, one has to destroy not only all holes,
but all ATs as well.  Using bounded search to fill holes and small
obstructions is now a pedestrian task
\cite{kaplan-99-chordal-completion,cai-96-hereditary-graph-modification,villanger-09-interval-completion};
which focuses us on large witnesses for ATs in a chordal graph.  Such
a witness is known to have the \emph{shallow} property.  For a shallow
witness $X$, there is a set of ${\cal O}(|X|)$ edges such that the
insertion of any of them will suffices to break $X$.  This hints its
disposal has to be harder than the holes; a similar dichotomy has been
observed in breaking holes by deletion versus by completion
\cite{marx-10-chordal-deletion}.

Observing the fixed-parameter tractability of {\sc chordal
  completion}, Kaplan et al.~\cite{kaplan-99-chordal-completion} asked
if the apparently harder {\sc interval completion} problem is FPT as
well.  This question was, after a dozen years, resolved by Villanger
et al.~\cite{villanger-09-interval-completion}, who designed a
$k^{2k}\cdot n^{O(1)}$ time algorithm.  It remains the only know FPT
algorithm for the problem.  The main purpose of this paper is to
propose a simple and improved FPT algorithm.

\begin{theorem}\label{thm:main-alg}
  There is a $6^k\cdot n^{O(1)}$ time algorithm for deciding whether
  or not there is a set of at most $k$ edges whose insertion makes an
  $n$-vertex graph $G$ an interval graph.
\end{theorem}

\paragraph{Our techniques.}  The main technical observation behind our
disposal of a shallow witness is its frame and shallow terminal.  A
shallow witness contains a unique AT; they are called \emph{terminals}
of this shallow witness.  Its \emph{frame} is defined to be the union
of the terminals as well as their neighbors; it consists of $6$ or $7$
vertices.  All other vertices of the shallow witness are inner
vertices of the longest defining path.  (See Figure~\ref{fig:at}.)
The terminal neither in nor adjacent to this path is the \emph{shallow
  terminal}.  We show there is a set of $5$ special edges in the frame
such that if an interval supergraph $\widehat G$ contains none of
them, then it contains no other edge in the frame,---such a frame is
already {finalized}.  Then an edge between the shallow terminal and a
vertex out of the frame has to be inserted.  
This observation suggests that we branch on either inserting one of
the $5$ edges to its frame, or assuming the frame will appear in
$\widehat G$ \emph{as is}.
In the last case, we put aside the shallow terminal and work on the
remaining part.  

Here comes the two crucial combinatorial results that justify this
partition.  All shallow terminals in a chordal graph without small
obstructions are well clustered, i.e., either similar or disconnected;
and the similar ones can be treated in the same way.  Formally, we
show those shallow terminals form a set of disjoint \emph{modules} (a
set of vertices with the same neighbors out of it); and more
importantly, such a module can never be broken: edges between it and
other vertices must be inserted in an \emph{all-or-none} manner.  They
together permit us to have a clear cut on any shallow witness whose
frame is fixed.

We then turn to next shallow witness and repeat this process.  After
it is exhausted, we are left with a partition of two disjoint interval
subgraphs.  A polynomial-time procedure will suffice to merge them and
finish the job.  We also observe that for small obstructions a $6$-way
branching will suffice, down from the trivial one that takes up to
$12$ ways.  These studies enable us to achieve the desired time
complexity.

We would like to call special attention to the preservation of modules
in interval completion, which should not be confused with the similar
fact for the deletion problem, where it comes as a trivial consequence
of hereditary property.  On one hand, this property definitely
benefits the further study of this problem; indeed, if we assume the
existence of a small solution, the graph necessarily has many modules.
On the other hand, such a property can be shown to hold in completion
problems to other graph classes and might be helpful.

\paragraph{Related work.} This time complexity asymptotically matches
that of the algorithm for \textsc{interval deletion}
\cite{cao-12-interval-deletion}, which is inherently harder than
\textsc{interval completion}.  Compared to the deletion problem, the
completion of holes is well understood, both combinatorially
\cite{cai-96-hereditary-graph-modification} and computationally
\cite{kaplan-99-chordal-completion, fomin-12-subexponential-fill-in,
  fomin-13-local-search-fill-in}.  To fill a large hole we need a
large number of edges, while comparatively, the removal of a single
vertex will suffice to break an arbitrarily large hole.  So holes pose
themselves as a much more significant trouble to the deletion problem.
This fact, unfortunately, leads some authors to believe an approach
for \textsc{interval deletion} can be trivially adapted for
\textsc{interval completion}.  This is nevertheless not the case, and
the completion problem has its own peculiar difficulties we have to
surmount.  


An easy fact is, by removing any vertex from a minimal forbidden
subgraph, we break this subgraph once and for all.  Noting that
interval graphs are hereditary, an \emph{interval deletion set} can
thus be viewed as a \emph{hitting set} for all minimal forbidden
subgraphs.  On the other hand, an edge inserted to fix an erstwhile
forbidden subgraph might introduce a new one.  Such side effects
arouse great bitterness, and require extreme care on each step.  One
might then be tempted to consider a set $E^+$ of edges that hits every
minimal forbidden subgraph (both ends in this subgraph) and each unit,
i.e., a subset of $E^+$, is free of side effects.\footnote{Some author
  did claim a result based on a falsified assumption that a sequence
  of ``safe'' edge sets fixing all forbidden subgraphs of $G$ makes an
  interval completion set, where a set $E_X$ of edges is called
  ``safe'' if $G+E_X$ contains no new forbidden subgraph.}  A second
thought tells us that nevertheless this additional
\emph{side-effect-free} condition is only a placebo, and
adds nothing to what counts as a proof.  Observe that ``neither
$G+E_1$ nor $G+E_2$ contains a new forbidden subgraph'' is not a
sufficient condition for ``$G + (E_1\cup E_2)$ contains no new
forbidden subgraph.''  Indeed, if this sort of argument might work, it
has to be something like ``every to-be-inserted-edge-unit is
side-effect-free to every intermediate graph.''  Or equivalently,
there is an ordered partition ($E_1,\dots,E_{\ell}$) of $E^+$ such
that for any $1\le i< \ell$, the insertion of $E_{i+1}$ will not
introduce a new forbidden subgraph to $G + (E_1\cup\dots\cup E_i)$ as
a side effect.  This argument has to be extremely complicated, if
doable at all: a single edge will break everything, even if every
previous edge serves its purpose faithfully and successfully.

 \begin{table}[ht]
   \begin{center}
     \begin{tabular}{c r r r}
       \hline
       {\bf Graph Class} ${\cal G}$  & {\sc ${\cal G}$ deletion} 
       & {\sc ${\cal G}$ completion} & kernel for {\sc ${\cal G}$ 
         completion}
       \\
       \hline 
       perfect &  W[2]-hard \cite{heggernes-11-perfect-deletion} 
       & open & open\\
       chordal &  $2^{{\cal O}(k \log{k})}$ \cite{cao-13-chordal-deletion} 
       & $2^{{\cal O}(\sqrt{k} \log{k})}$ \cite{fomin-12-subexponential-fill-in}
       & $k^2$ \cite{natanzon-00-approximate-fill-in}
       \\
       interval & $10^k$ \cite{cao-12-interval-deletion}
       & $6^k$ [This paper] & open\\
       proper interval &  $6^k$ \cite{villanger-13-pivd} &  
       $16^k$ \cite{kaplan-99-chordal-completion} &
       $k^3$ \cite{bessy-11-kernels-proper-interval-completion}
       \\
       \hline
     \end{tabular}
   \end{center}
   \caption{Graph deletion/completion problems to graph classes}
   \label{fig:list-of-problems}
\end{table}

Interval graphs and chordal graphs are not the only graph classes that
receive attention in this respect.  Other graph classes include
perfect graphs and proper interval graphs.  A graph is \emph{perfect}
if neither it or its complement contains an odd hole.  An interval
graph is a \emph{proper interval graph} if it has a representation
with no interval containing another one.  These four classes have a
proper containment relations, proper interval graphs are a subclass of
interval graphs, and all graph classes are perfect.
The known results on these graphs classes are summarized in
Table~\ref{fig:list-of-problems}.  Note that a W[2]-hard problem is
unlikely to be FPT \cite{downey-fellows-99}.

This paper is organized as follows.  Section~\ref{sec:pre} sets the
definitions and recalls some known facts.
Section~\ref{sec:forbidden-subgraphs} depicts minimal ways to fix
minimal forbidden graphs.  Section~\ref{sec:module} characterizes the
external and internal behavior of modules in a minimum interval
supergraph.  Section~\ref{sec:shallow} investigates shallow terminals
in chordal graphs with no small obstructions.  Section~\ref{sec:alg}
presents our bounded search tree algorithm, and proves
Theorem~\ref{thm:main-alg}.  \tagged{full}{ Section~\ref{sec:remark}
  closes this paper with some technical remarks.  }

\section{Preliminaries}\label{sec:pre}

Graphs discussed in this paper shall always be undirected and simple.
The vertex set and edge set of a graph $G$ are denoted by $V(G)$ and
$E(G)$ respectively.  The \emph{size} $||G||$ of graph $G$ is defined
to be $|E(G)|$, i.e., the number of edges in it.  We say $G'$ is a
\emph{supergraph} of $G$ if $V(G) = V(G')$ and $E(G)\subseteq E(G')$.
Given a set $E'$ of \emph{missed} edges in $V(G)^2\setminus E(G)$, the
supergraph $G+E'$ is defined to be $(V(G), E(G)\cup E')$.

We say that a pair of vertices $u$ and $v$ is \emph{adjacent} (to each
other) if they are connected by an edge, denoted by $v_1 \sim v_2$;
otherwise \emph{nonadjacent} and denoted by $v_1 \not\sim v_2$.  Two
vertex sets $X$ and $Y$ are \emph{completely connected} if $x \sim y$
for every pair of $x \in X$ and $y \in Y$.  We denoted by $N_G(v)$ by
the set of \emph{neighbors} of $v$, i.e., vertices adjacent to $v$,
and $N_G[v] = N_G(v) \cup \{v\}$.  Let $U\subseteq V(G)$ be a subset
of vertices.  Neighbors of $U$ are defined analogously: $N_G[U] =
\bigcup_{v \in U} N_G[v]$ and $N_G(U) = N_G[U] \backslash U$.  The
subscript will be omitted if it is clear from context.  A
\emph{clique} in a graph where every pair of vertices is adjacent.  A
vertex is \emph{simplicial} if its neighbors induce a clique.  The
subgraph induced by $U$ is denoted by $G[U]$, and $G - U$ is used as a
shorthand for $G[V(G) \setminus U]$.  We say a connected induced
subgraph $G[U]$ is a \emph{connected component} of $G$ if $N(U) =
\emptyset$.

A sequence of \emph{distinct} vertices $v_0 v_1 \dots v_l$ such that
$v_i \sim v_{i+1}$ for each $0 \leq i < l$ is called a \emph{path}, or
a \emph{$v_0$-$v_l$ path} when the two end vertices are of special
interest; the \emph{length} is defined to be $l$.  If $l > 1$ and $v_0
\sim v_l$, then the sequence $v_0 v_1 \dots v_l v_0$ is called a
\emph{cycle} of \emph{length} $l + 1$.  As an abuse of notation, we
use $u \in P$ (resp., $u \in C$) to denote that the vertex $u$ appears
in the path $P$ (resp., cycle $C$), i.e., we also consider a path
(resp., cycle) as the set of elements in the sequence.  A \emph{hole}
is an induced cycle of length at least $4$.

An interval representation of an interval graph $G$ is given by ${\cal
  I} = \{I_v | v\in V(G)\}$.  Each vertex $v$ corresponds to an
\emph{closed} interval $I_v$ with endpoints $\emph{left}(v)$ and
$\emph{right}(v)$, respectively; i.e., $I_v = [\emph{left}(v),
\emph{right}(v)]$.  Likewise, for a subset $X$ of vertices, we define
$\emph{left}(X) = \min_{x\in X} \emph{left}(x)$ and $\emph{right}(X) =
\max_{x\in X} \emph{right}(x)$.  Observe that if a subset $X$ of
vertices induces a connected subgraph, then the union of $\{I_v|v\in
X\}$ also forms an interval.  Assume without loss of generality, no
intervals for distinct vertices share a same endpoint; as a result, an
interval representation of graph $G$ defines precisely $2 |V(G)|$
distinct endpoints.  For any point $p$, we can find a small positive
value $\epsilon$ such that in $[p-\epsilon, p+\epsilon]$ the only
possible endpoint of an interval for some vertex is $p$.  We define
$K_p = \{v| p\in I_v\}$, which clearly induces a clique.

A set $X$ of vertices is an \emph{$x$-$y$ separator} if $x,y\not\in X$
and there is no $x$-$y$ path in $G -X$.  The following lemma relates
an interval representation of an interval graph and its separators.
\begin{lemma}\label{thm:interval-representation}
  Let $u$ and $v$ be a pair of nonadjacent vertices in a connected
  interval graph $G$, and $\cal I$ be an interval representation of
  $G$.  A set $X$ of vertices is a $u$-$v$ separator if and only if
  $K_p\subseteq X$ for some point $p$ such that $I_u$ and $I_v$ lie on
  different sides of $p$.
\end{lemma}
We say that an interval graph $\widehat G$ is an \emph{interval
  supergraph} of $G$ if $G\subseteq\widehat G$.  An interval
supergraph $\widehat G$ is {\em minimum} if there is no strictly
smaller interval supergraph of $G$.  No generality will be lost by
assuming $G$ is connected.  The problem is hence defined as follows:
\begin{quote}
  {\sc interval completion}: ($G,k$): Given a connected graph $G$ and
  a nonnegative integer $k$, find an interval supergraph of $G$ of
  size no more than $||G|| + k$, or report no such a graph exists.
\end{quote}
Immediately from the definition we have the following observation.
\begin{lemma}\label{lem:partial-solution}
  A graph $\widehat G$ is a minimum interval supergraph of $G$ if and
  only if it is a minimum interval supergraph of any graph $\widehat
  G'$ satisfying $G \subseteq \widehat G' \subseteq \widehat G$.
\end{lemma}


\section{Forbidden induced subgraphs}
\label{sec:forbidden-subgraphs}

A \emph{forbidden induced subgraph} refers to a non-interval graph,
and it is minimal if every proper induced subgraph of it is an
interval graph.  Three vertices form an \emph{asteroidal triple} (AT)
if each pair of them is connected by a path that avoids the neighbors
of the third one.  We use \emph{asteroidal witness} (AW) to refer to a
minimal forbidden induced subgraph that is not a hole.  It should be
easy to check that an AW contains precisely one AT,
called \emph{terminals} of this AW; and its vertex set is the union of
these three defining paths for this triple.  By definition, the
terminals are the only simplicial vertices of this AW and they are
nonadjacent to each other.  Lekkerkerker and Boland
\cite{lekkerkerker-62-interval-graphs} observed that a graph is an
interval graph if and only if it is chordal and contains no AW, and
more importantly, proved the following characterization.

\begin{theorem}[\cite{lekkerkerker-62-interval-graphs}]
  A minimal forbidden induced subgraph is either a hole or an AW
  depicted in Figure~\ref{fig:at}.
\end{theorem}

\tikzstyle{corner}  = [fill=blue,inner sep=3pt]
\tikzstyle{special} = [fill=black,circle,inner sep=2pt]
\tikzstyle{vertex}  = [fill=black,circle,inner sep=2pt]
\tikzstyle{edge}    = [draw,thick,-]
\tikzstyle{at edge} = [draw,ultra thick,-,red]
\tikzstyle{fill edge} = [dashed,red,bend right]
\begin{figure*}[ht]
  \centering
  \begin{subfigure}[b]{0.22\textwidth}
    \centering
    \begin{tikzpicture}[scale=.2]
    \node [corner,label=right:$t_1$] (s) at (0,6.44) {};
    \node [corner,label=above:$t_2$] (a) at (-7, 0) {};
    \node [vertex] (a1) at (-4, 0) {};
    \node [special,label=45:$c$] (v) at (0, 0) {};
    \node [vertex] (b1) at (4, 0) {};
    \node [corner,label=above:$t_3$] (b) at (7, 0) {};
    \node [vertex] (c) at (0,3.5) {};
    \draw[] (a) -- (a1) -- (v) -- (b1) -- (b);
    \draw[] (v) -- (c) -- (s);
    \end{tikzpicture}
    \caption{long claw}
    \label{fig:long-claw}
  \end{subfigure}%
  \qquad
  \begin{subfigure}[b]{0.22\textwidth}
    \centering
    \begin{tikzpicture}[scale=.2]
    \node [corner,label=right:$t_1$] (s) at (0,3.2) {};
    \node [corner,label=above:$t_2$] (a) at (-7, 0) {};
    \node [vertex] (a1) at (-4, 0) {};
    \node [special,label=45:$c$] (v) at (0, 0) {};
    \node [vertex] (b1) at (4, 0) {};
    \node [corner,label=above:$t_3$] (b) at (7, 0) {};
    \node [vertex] (c) at (0,-3.8) {};
    \draw[] (a) -- (a1) -- (v) -- (b1) -- (b) -- (c) -- (a);
    \draw[] (b1) -- (c) -- (a1);
    \draw[] (v) -- (c) -- (s);
    \end{tikzpicture}
    \caption{whipping top}
    \label{fig:top}
  \end{subfigure}%

  \begin{subfigure}[b]{0.45\textwidth}
    \centering
    \begin{tikzpicture}[scale=.45]
    \node [corner,label=right:$s$] (s) at (0,4.44) {};
    \node [corner,label=above:$l$,label=below:$b_{0}$] (a) at (-7, 0) {};
    \node [special,label=above:$h$,label=below:$b_{1}$] (a1) at (-5, 0) {};
    \node [vertex,label=below:$b_{2}$] (a2) at (-3, 0) {};
    \node [vertex,label=below:$b_{i}$] (bi) at (0, 0) {};
    \node [vertex,label=below:$b_{d-1}$] (b2) at (3, 0) {};
    \node [special,label=above:$t$,label=below:$b_d$] (b1) at (5, 0) {};
    \node [corner,label=above:$r$,label=below:$b_{d+1}$] (b) at (7, 0) {};
    \node [special,label=5:$c$] (c) at (0,2.7) {};
    \draw[] (a) -- (a1) -- (a2) (b2) -- (b1) -- (b);
    \draw[] (bi) -- (c) -- (s);
    \draw[] (a1) -- (c) -- (b1);
    \draw[] (a2) -- (c) -- (b2);
    \draw[dashed] (a2) -- (b2);
    \end{tikzpicture}
    \caption{$\dag$-AW ($d\ge 2$)}
    \label{fig:dag}
  \end{subfigure}%
  \begin{subfigure}[b]{0.45\textwidth}
    \centering
    \begin{tikzpicture}[scale=.4]
    \node [corner,label=right:$s$] (s) at (0,4.44) {};
    \node [corner,label=above:$l$,label=below:$b_{0}$] (a) at (-7, 0) {};
    \node [special,label=above:$h$,label=below:$b_{1}$] (a1) at (-5, 0) {};
    \node [vertex,label=below:$b_{2}$] (a2) at (-3, 0) {};
    \node [vertex,label=below:$b_{i}$] (bi) at (0, 0) {};
    \node [vertex,label=below:$b_{d-1}$] (b2) at (3, 0) {};
    \node [special,label=above:$t$,label=below:$b_d$] (b1) at (5, 0) {};
    \node [corner,label=above:$r$,label=below:$b_{d+1}$] (b) at (7, 0) {};
    \node [special,label=left:$c_1$] (c1) at (-1.5,2.7) {};
    \node [special,label=right:$c_2$] (c2) at (1.5,2.7) {};
    \draw[] (a) -- (a1) -- (a2) (b2) -- (b1) -- (b);
    \draw[] (bi) -- (c1) -- (s) -- (c2) -- (bi);
    \draw[] (a) -- (c1) -- (c2) -- (b);
    \draw[] (a1) -- (c1) -- (b1) -- (c2) -- (a1);
    \draw[] (a2) -- (c1) -- (b2) -- (c2) -- (a2);
    \draw[dashed] (a2) -- (b2);
    \end{tikzpicture}
    \caption{$\ddag$-AW ($d\ge 1$)}
    \label{fig:ddag}
  \end{subfigure}%
  \caption{Asteroidal witnesses in a chordal graph (terminals are
    marked as squares).}
  \label{fig:at}
\end{figure*}
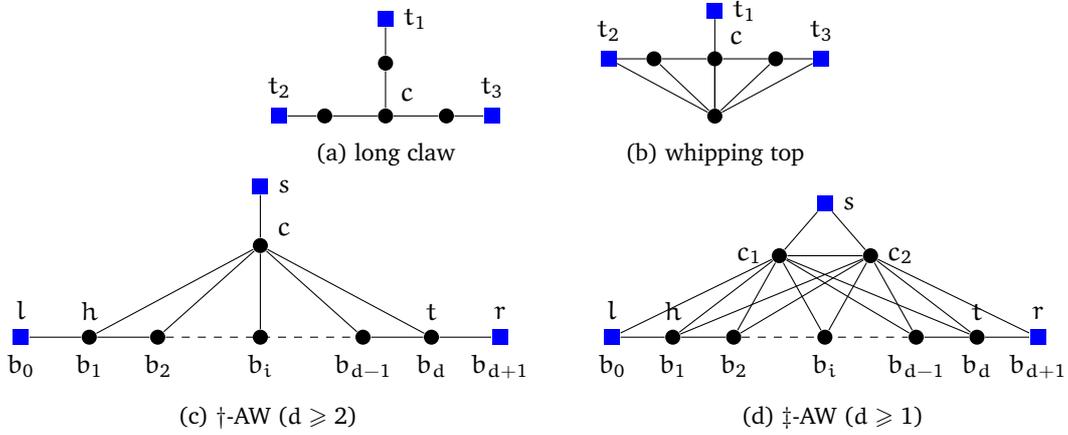

Some remarks are in order.  First, it is easy to verify that a hole of
$6$ or more vertices witnesses an AT, e.g., any three
nonadjacent vertices within it, but following convention, we only
refer to it as a hole, and reserve the term AW for graphs listed in
Figure~\ref{fig:at}.  Second, for the purpose of the current paper, we
single out $\dag$- and $\dag$-AWs with $d\le 3$, and denote them by
\emph{$d$-nets}, and respectively \emph{$d$-tents}; they, together
with long claws and whipping tops, are called \emph{small AWs}.  The
others, i.e., $\dag$- and $\dag$-AWs with $d> 3$, are called
\emph{long AW}.  

The \emph{frame} of a long AW is defined to be the union of the
terminals and their neighbors.  By definition, at least one vertex of
the long AW is not in its frame; all of them belong to the longest
defining path.  
The ends ($l,r$) and inner vertices ($B = \{b_1, \dots, b_{d}\}$) of
this path are called \emph{base terminals} and \emph{base (vertices)}
respectively.
The other terminal $s$ is the \emph{shallow terminal}, whose
neighbor(s) $c$ or ($c_1,c_2$) are the \emph{center(s)}.

To avoid repetition of the essentially same argument for $\dag$-AWs
and $\ddag$-AWs, we use a generalized notation for both $\dag$- and
$\ddag$-AWs.  In particular, both $c_1$ and $c_2$ refer to the only
center $c$ when it is a $\dag$-AW.  As long as we do not use the
adjacency of $c_1$ and $l$, $c_2$ and $r$, or $c_1$ and $c_2$ in any
of the arguments, this unified (abused) notation will not introduce
inconsistencies.\footnote{Albeit the frames in $\dag$- and $\ddag$-AWs
  have different number of vertices and different number of edges,
  they both have $10$ missed edges.  Moreover, using our generalized
  notations, these edges are exactly the same.  See
  Lemma~\ref{lem:fill-long-aw} and its corollary.} For the sake of
notational convenience, we will use $h$ and $t$ to refer to the base
vertices $b_1$ and $b_{d}$, respectively; the frame is then denoted by
$(s: c_1,c_2: l,h; t,r)$.

\begin{lemma}\label{lem:only-terminals-are-simplicial}
  In an AW a vertex is simplicial if and only if it is a terminal. 
\end{lemma}

In time ${\cal O}(n^{5})$, we can find a minimal forbidden induced
subgraph or asserts its nonexistence as follows.  For a hole, we guess
three consecutive vertices $\{h_1,h_2,h_3\}$, and then search for the
shortest $h_1$-$h_3$ path in $G - (N[h_2]\setminus \{h_1,h_3\})$.  For
an AW, we guess three independent vertices $\{t_1,t_2,t_3\}$, and for
$i=1,2,3$, search for the shortest path between the other two in $G -
N[t_i]$.  Since the AW found as such is the minimum among what
witnesses the AT $\{t_1,t_2,t_3\}$, in the same time we
can actually construct a small AW or asserts its nonexistence.

It is now well known that holes can be easily filled in.  A long hole
of more than $k+3$ vertices will immediately imply ``NO,'' while a
short hole has only a bounded number of minimal ways to fill
\cite{kaplan-99-chordal-completion,cai-96-hereditary-graph-modification},
of which an interval supergraph $\widehat G$ of $G$ must contain one.

\begin{lemma}\label{lem:holes}
  A minimal set of edges that fill a hole $H$ has size $|H| - 3$, and
  the number of such sets is upper bounded by $4^{|H|-2}$.
\end{lemma}

Given any AW $W$, unless we insert an edge between one terminal and
the defining path connecting the other terminals, the terminals will
remain an AT; witnessed by a subset of $W$.  This
speedily produces a set of ${\cal O}(|W|)$ edges $E_W$ that any
interval supergraph cannot avoid; the number is $d+7$ and $d+5$ for a
$\dag$- and $\ddag$-AW respectively.
 However, the insertion of some edge in $E_W$, e.g., an edge between
 two terminals, might bring hole(s) to $W$, which in turn demands the
 insertion of a set of edges stated in {Lemma}~\ref{lem:holes}.  Upon
 a closer scrutiny, one sees it can be done slightly more efficiently.
 If there exists a smaller set $E'_W$ of edges than $E_W$ such that
 for each edge $e\in E_W\setminus E'_W$, we have to insert at least
 one edge of $E'_W$ after the insertion of $e$, then we may branch on
 insertion of edges in $E'_W$ instead.  Here we do not require $E'_W$
 to be a subset of $E_W$.  Specifically for long AWs we have (noting
 $d>3$)

\begin{lemma}\label{lem:fill-long-aw}
  Let $(s: c_1,c_2: l,h;t,r)$ and $\{b_1,\dots,b_d\}$ be the frame and
  base of a long AW in $G$.  Any interval supergraph $\widehat G$ of
  $G$ must contain one of the $d+3$ edges:
  \begin{equation}
    \label{eq:E_W}
    \{l c_2, c_1 r, h t, s h, s t\} \cup \{s b_i | 1< i< d\}.    
  \end{equation}
  Moreover, if $\widehat G$ contains none of $\{l c_2, c_1 r, h t, s
  h, s t\}$, then the frame induces the same subgraph in $G$ and
  $\widehat G$.
\end{lemma}
\begin{proof}
  We may prove the second assertion first.  
  We show $\widehat G$ cannot contain any of the other $5$ missed
  edges $\{s l, s r, l r, l t, h r\}$ in the frame.  The insertion of
  edge $s l$ will introduce a $4$-hole $(s c h l s)$, which requires
  the insertion of at least one of $s h$ and $l c_2$, which are not
  allowed.  A symmetric argument apply to the edge $s r$.  The
  insertion of edge $l t$ will introduce a $4$-hole $(l h c_2 t l)$,
  which requires the insertion of at least one of $h t$ and $l c_2$,
  which are not allowed.  A symmetric argument apply to the edge $h
  r$.  Now we are left with only $l r$, whose insertion will introduce
  a $5$-hole $(c h l r t c)$ or $4$-hole $(l c_1 c_2 r l)$ depending
  on the type of the AW.  Every minimal set of edges that fill this
  hole need edges that have been excluded.

  To make the terminals cease to form an AT, we have to insert an edge
  to connect one terminal and the defining path connecting other
  terminals.  This makes a list of at most $d+7$ edges, among which
  only $\{s l, s r, l r, l t, h r\}$ are not included in
  (\ref{eq:E_W}).  All of them are in the frame, and hence cannot be
  inserted without also inserting one of (\ref{eq:E_W}).  This proves
  the first assertion and completes the proof.
\end{proof}

The exclusion of edges $\{l c_2, c_1 r, h t, s h, s t\}$ will
perpetuate the structure of the frame, which is hence called an
\emph{unchangeable frame}.  In most case, the frame is the only
structure in a long AW that concerns us.

The purpose of {Lemma}~\ref{lem:fill-long-aw} is surely not trying to
decrease the directions we need to branch on a long AW, from $d+7$ or
$d+5$ to $d+3$.  Instead, we are after the structural information that
turns out to be crucial.
Of special significance is the edge $h t$; though its insertion does
not break this AT.  This is formulated in the following
corollary and visualized in Figure~\ref{fig:fixed-frame}.  Noting that
as a consequence of $h\not\sim t$, intervals $I_h$ and $I_t$ are
disjoint.

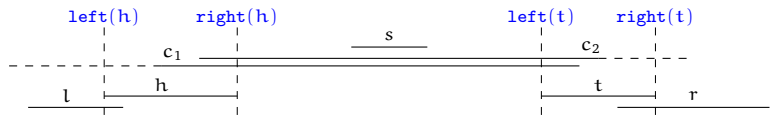
\begin{figure*}[h!]
  \centering
  \begin{tikzpicture}[scale=.5]
    \scriptsize
    \draw  (9.5,1.5) -- (13,1.5);  \node at (11,1.8) {$h$};
    \draw  (7.5,1.2) -- (10,1.2);   \node at (8.5,1.5) {$l$};
    \draw  (23,1.2) -- (27,1.2); \node at (25,1.5) {$r$};
    \draw (21,1.5) -- (24,1.5);  \node at (22.5,1.8) {$t$};
    
    \draw[dashed] (7, 2.3) -- (11, 2.3); \draw(11, 2.3)-- (22, 2.3); \node
    at (11.3, 2.6) {$c_1$};
    \draw (12,2.5) -- (22.5,2.5); \draw[dashed] (22.5,2.5) --(25,2.5);
    \node at (22.3, 2.8) {$c_2$};
    
    \draw  (16, 2.8) -- (18, 2.8);  \node at (17, 3.1) {$s$};
    
    \draw[ultra thin,dashed] (9.5,1) -- (9.5,3.5);\node[blue] at (9.5,3.6)
    {$\texttt{left}(h)$};
    \draw[ultra thin,dashed] (13,1) -- (13,3.5);\node[blue] at (13,3.6)
    {$\texttt{right}(h)$};
    \draw[ultra thin,dashed] (21,1) -- (21,3.5);\node[blue] at (21,3.6)
    {$\texttt{left}(t)$};
    \draw[ultra thin,dashed] (24,1) -- (24,3.5);\node[blue] at (24,3.6)
    {$\texttt{right}(t)$};
  \end{tikzpicture}
  \caption{Interval rep. of a unchangeable frame.}
  \label{fig:fixed-frame}
\end{figure*}
\begin{corollary}\label{lem:fill-long-aw-2}
  Let $(s: c_1,c_2: l,h;t,r)$ be the frame of a long AW in graph $G$.
  If an interval supergraph $\widehat G$ of $G$ contains none of $\{l
  c_2, c_1 r, h t, s h, s t\}$, then in any interval representation of
  $\widehat G$, the interval for $s$ is between intervals for $h$ and
  $t$.  
\end{corollary}
A similar and straightforward check furnishes for each small AW a set
of at most $6$ edges, reduced from e.g., 12 candidate edges for a
long-claw.  The proof is deferred to the appendix for lack of space.

\begin{lemma}\label{lem:6-enough}
  For each small AW in a graph $G$, there is a set of at most $6$
  edges, depicted as dashed edges in Figure~\ref{fig:fill-edges}, such
  that any interval supergraph of $G$ contains at least one of them.
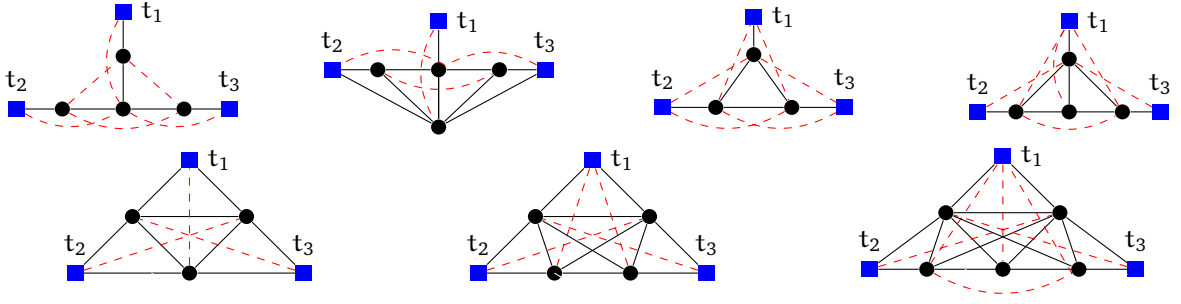
\begin{figure*}[ht]
  \centering
  \begin{subfigure}[b]{0.23\textwidth}
    \centering
    \begin{tikzpicture}[scale=.2]
    \node [corner,label=right:$t_1$] (s) at (0,6.44) {};
    \node [corner,label=above:$t_2$] (a) at (-7, 0) {};
    \node [vertex] (a1) at (-4, 0) {};
    \node [special] (v) at (0, 0) {};
    \node [vertex] (b1) at (4, 0) {};
    \node [corner,label=above:$t_3$] (b) at (7, 0) {};
    \node [vertex] (c) at (0,3.5) {};
    \draw[] (a) -- (a1) -- (v) -- (b1) -- (b);
    \draw[] (v) -- (c) -- (s);

    \draw[fill edge] (s) to (v) (a) to (v);
    \draw[fill edge] (v) to (b) (a1) to (b1);
    \draw[fill edge] (a1) -- (c) -- (b1);
    \end{tikzpicture}
  \end{subfigure}%
  \quad
  \begin{subfigure}[b]{0.23\textwidth}
    \centering
    \begin{tikzpicture}[scale=.2]
    \node [corner,label=right:$t_1$] (s) at (0,3.2) {};
    \node [corner,label=above:$t_2$] (a) at (-7, 0) {};
    \node [vertex] (a1) at (-4, 0) {};
    \node [special] (v) at (0, 0) {};
    \node [vertex] (b1) at (4, 0) {};
    \node [corner,label=above:$t_3$] (b) at (7, 0) {};
    \node [vertex] (c) at (0,-3.8) {};
    \draw[] (a) -- (a1) -- (v) -- (b1) -- (b) -- (c) -- (a);
    \draw[] (b1) -- (c) -- (a1);
    \draw[] (v) -- (c) -- (s);

    \draw[fill edge] (s) to (c);
    \draw[fill edge] (b) to (v) to (a) (a1) to (b1);
    \end{tikzpicture}
  \end{subfigure}%
  \quad
  \begin{subfigure}[b]{0.23\textwidth}
    \centering
    \begin{tikzpicture}[scale=.2]
    \node [corner,label=right:$t_1$] (s) at (0,6) {};
    \node [corner,label=above:$t_2$] (a) at (-6, 0) {};
    \node [vertex] (a1) at (-2.5, 0) {};
    \node [vertex] (b1) at (2.5, 0) {};
    \node [corner,label=above:$t_3$] (b) at (6, 0) {};
    \node [vertex] (c) at (0,3.5) {};
    \draw[] (a) -- (a1) -- (b1) -- (b);
    \draw[] (c) -- (s);
    \draw[] (a1) -- (c) -- (b1);

    \draw[fill edge] (a) to (b1) (a1) to (b);
    \draw[fill edge] (a1) -- (s) -- (b1);
    \draw[fill edge] (a) -- (c) -- (b);
    \end{tikzpicture}
  \end{subfigure}%
  \quad
  \begin{subfigure}[b]{0.23\textwidth}
    \centering
    \begin{tikzpicture}[scale=.2]
    \node [corner,label=right:$t_1$] (s) at (0,6) {};
    \node [corner,label=above:$t_2$] (a) at (-6, 0) {};
    \node [vertex] (a1) at (-3.5, 0) {};
    \node [vertex] (u) at (0, 0) {};
    \node [vertex] (b1) at (3.5, 0) {};
    \node [corner,label=above:$t_3$] (b) at (6, 0) {};
    \node [vertex] (c) at (0,3.5) {};
    \draw[] (a) -- (a1) -- (b1) -- (b);
    \draw[] (u) -- (c) -- (s);
    \draw[] (a1) -- (c) -- (b1);

    \draw[fill edge] (a1) to (b1) (s) to (u);
    \draw[fill edge] (a1) -- (s) -- (b1);
    \draw[fill edge] (a) -- (c) -- (b);
    \end{tikzpicture}
  \end{subfigure}%

  \begin{subfigure}[b]{0.3\textwidth}
    \centering
    \begin{tikzpicture}[scale=.25]
    \node [corner,label=right:$t_1$] (s) at (0,6) {};
    \node [corner,label=above:$t_2$] (a) at (-6, 0) {};
    \node [vertex] (a1) at (0, 0) {};
    \node [corner,label=above:$t_3$] (b) at (6, 0) {};
    \node [vertex] (c1) at (-3,3) {};
    \node [vertex] (c2) at (3,3) {};
    \draw[] (a) -- (a1) -- (b) -- (c2) -- (s) -- (c1) -- (a);
    \draw[] (c1) -- (c2) -- (a1) -- (c1);

    \draw[fill edge] (a1) -- (s) (b) -- (c1) (a) -- (c2);
    \draw[white, bend right] (-2,0)  to (b);
    \end{tikzpicture}
  \end{subfigure}%
  \quad
  \begin{subfigure}[b]{0.3\textwidth}
    \centering
    \begin{tikzpicture}[scale=.25]
    \node [corner,label=right:$t_1$] (s) at (0,6) {};
    \node [corner,label=above:$t_2$] (a) at (-6, 0) {};
    \node [vertex] (a1) at (-2, 0) {};
    \node [vertex] (b1) at (2, 0) {};
    \node [corner,label=above:$t_3$] (b) at (6, 0) {};
    \node [vertex] (c1) at (-3,3) {};
    \node [vertex] (c2) at (3,3) {};
    \draw[] (a) -- (a1) -- (b) -- (c2) -- (s) -- (c1) -- (a);
    \draw[] (c1) -- (c2) -- (a1) -- (c1);

    \draw[white, bend right] (-2,0)  to (b); 
    \draw[fill edge] (b1) -- (s);
    \draw[fill edge] (s) -- (a1) (b) -- (c1) (a) -- (c2);
    \draw (c1) -- (b1) -- (c2);
  \end{tikzpicture}
  \end{subfigure}
  \quad
  \begin{subfigure}[b]{0.3\textwidth}
    \centering
    \begin{tikzpicture}[scale=.25]
    \node [corner,label=right:$t_1$] (s) at (0,6) {};
    \node [corner,label=above:$t_2$] (a) at (-7, 0) {};
    \node [vertex] (a1) at (-4, 0) {};
    \node [vertex] (u) at (0, 0) {};
    \node [vertex] (b1) at (4, 0) {};
    \node [corner,label=above:$t_3$] (b) at (7, 0) {};
    \node [vertex] (c1) at (-3,3) {};
    \node [vertex] (c2) at (3,3) {};
    \draw[] (a) -- (a1) -- (b) -- (c2) -- (s) -- (c1) -- (a);
    \draw[] (c1) -- (c2) -- (a1) -- (c1) -- (u) -- (c2);

    \draw[white, bend right] (-2,0)  to (b); 
    \draw[fill edge] (a1) to (b1) -- (s) -- (u);
    \draw[fill edge] (s) -- (a1) (b) -- (c1) (a) -- (c2);
    \draw (c1) -- (b1) -- (c2);
  \end{tikzpicture}
  \end{subfigure}
  \caption{A minimum interval supergraph must contain a dashed edges.}
  \label{fig:fill-edges}
\end{figure*}
\end{lemma}

One should compare the edges used for $3$-nets and $3$-tents with
those for long AWs in Lemma~\ref{lem:fill-long-aw}.  It is worth
noting that the threshold for the base length of a long AW, $d\ge 4$,
is carefully chosen.  A partial motivation is the following lemma.
See also the proof of Lemma~\ref{lem:shallow-1} (Section
~\ref{sec:shallow-1} ) for more motivations and
Section~\ref{sec:remark} for a discussion.

\begin{lemma}\label{lem:common-neighbor-of-base}
  Let $X = (s: c_1,c_2: l,h;t,r)$ be the frame of a long AW $W$ in
  graph $G$, and $x$ be a neighbor of both $h$ and $t$.  If $x\not\sim
  s$, then there is a small AW in $G[X\cup \{x\}]$, and any interval
  supergraph of $G$ contains at least one of $\{l c_2, c_1 r, s h, s
  t, h t\}$ or $s x$.
\end{lemma}
\begin{proof}
  From the adjacencies between $x$ and $\{h,t,s\}$ it can be inferred
  $x\not\in W$.  There is an $l$-$r$ path $(l h x t r)$ that avoids
  $N[s]$, which replaces $(l B r)$ in $W$ to define a small AW.  (See
  Figure~\ref{fig:common-neighbor}, where the case $N(x)\cap \{l,r\} =
  l$ is symmetric to the case $N(x)\cap \{l,r\} = r$.)  The second
  assertion ensues.
\end{proof}

\section{Modules}\label{sec:module}
A subset $M$ of vertices forms a \emph{module} of $G$ if all vertices
in $M$ have the same neighbors outside of $M$.  In other words, for
any pair of vertices $u,v \in M$, a vertex $x \not\in M$ is either
adjacent to both or neither of $u$ and $v$.
A brief inspection reveals that none of graphs in Figure~\ref{fig:at}
has a module $M$ satisfying $1<|M| < |V(G)|$, and this is true also
for holes of length greater than 4.  Recall that a minimal forbidden
induced subgraph contains at least $4$ vertices.
\begin{lemma}\label{lem:module-exchangeable}
  Let $M$ be a module of graph $G$.  If a subset $X$ of vertices
  induces a minimal forbidden subgraph, then either $X \subseteq M$,
  or $|M \cap X| \leq 2$, where equality only holds if $X$ induces a
  $4$-hole.
\end{lemma}
Naturally, one will surmise that a minimum interval supergraph
$\widehat G$ of $G$ will preserve modules of $G$.  Nonetheless, this
does not hold true in general.  Here we manage to show a slightly
weaker version; by \emph{connected module} we mean a module that
induces a connected subgraph.  The way we prove this crucial theorem
is by working on an interval representation $\cal I$ of the interval
supergraph $\widehat G$: if $M$ is not a module of $\widehat G$, then
we modify $\{I_v|v\in M\}$ to make a new interval representation $\cal
I'$ that gives a strictly smaller interval supergraph $\widehat G'$ of
$G$.  In one case we will use the \emph{project} operation defined as
follows.  Given a set $X$ of sub-intervals of $[p, q]$ and another
interval $[p', q']$, we \emph{project} $x\in X$ to $x'$ by applying
the mapping $$\imath\to \frac{q'-p'}{q - p}(\imath-p) + p'$$ to both
endpoints of $x$.  It is easy to verify that $x'$ is a sub-interval of
$[p', q']$, and more importantly, these two sets of intervals
represent the same interval graphs.

\begin{theorem}\label{thm:preserving-modules}
  A connected module $M$ of graph $G$ remains a module in any minimum
  interval supergraph of $G$.
\end{theorem}
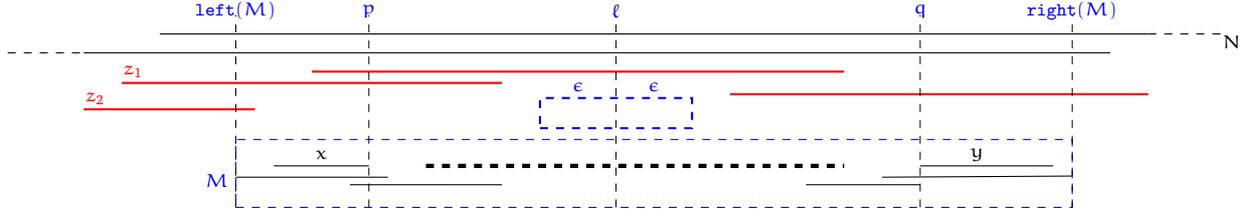
\begin{figure*}[h!]
  \centering
\begin{tikzpicture}[scale=.5]
\scriptsize

\draw  (6,1.2) -- (10,1.2);  
\draw  (7,1.5) -- (9.5,1.5);  \node at (8.25,1.8) {$x$};
\draw  (9,1) -- (13,1);  
\draw  (21,1) -- (24,1);  
\draw  (23,1.2) -- (28,1.23);
\draw (24,1.5) -- (27.5,1.5);  \node at (25.5,1.8) {$y$};
\draw[ultra thick,dashed] (11,1.5) -- (22,1.5);
\draw[blue,dashed] (6,0.4) -- (28,0.4) -- (28,2.2) -- (6,2.2) -- (6,0.4);  
\node[blue] at (5.5,1.1) {$M$};

\draw[thick,red]  (2,3) -- (6.5,3) (3,3.7) -- (13,3.7) (8,4) -- (22,4) (19,3.4) -- (30,3.4);
\node[red] at (2.3,3.3) {$z_2$}; \node[red] at (3.3,4) {$z_1$};  

\draw[dashed] (0,4.5) -- (2,4.5);  \draw  (2,4.5) -- (29,4.5);  
\draw  (4,5) -- (30,5);  \draw[dashed] (30,5) -- (32,5);  
\node[thick] at (32.2,4.8) {$N$};

\draw[ultra thin,dashed] (9.5,0.4) -- (9.5,5.5);\node[blue] at (9.5,5.6) {$p$};
\draw[ultra thin,dashed] (24,0.4) -- (24,5.5);\node[blue] at (24,5.6) {$q$};
\draw[ultra thin,dashed] (6,0.4) -- (6,5.5);\node[blue] at (6,5.6) {$\texttt{left}(M)$};
\draw[ultra thin,dashed] (28,0.4) -- (28,5.5);\node[blue] at (28,5.6) {$\texttt{right}(M)$};
\draw[ultra thin,dashed] (16,0.4) -- (16,5.5);\node[blue] at (16,5.6) {$\ell$};
\draw[blue,thick,dashed] (14,2.5) -- (18,2.5) -- (18,3.3) -- (14,3.3) -- (14,2.5);
\node[blue] at (15,3.6) {$\epsilon$};\node[blue] at (17,3.6) {$\epsilon$};
\end{tikzpicture}
  \caption{Interval rep. of module $M$ and $N(M)$.\\
  A thick (red) interval breaks the modularity of $M$.}
  \label{fig:interval-representation}
\end{figure*}
\begin{proof}
  Without loss of generality, we may assume $|M| > 1$ and
  $N_G(M)\ne\emptyset$.  Let ${\widehat G} = (V, \widehat E)$ be a
  minimum interval supergraph of $G$ and ${\cal I} = \{I_v|v\in V\}$
  be an interval representation of ${\widehat G}$.  We define $p =
  \min_{v\in M}\mathtt{right}(v)$ and $q = \max_{v\in
    M}\mathtt{left}(v)$; in particular $\mathtt{right}(x) = p$ and
  $\mathtt{left}(y) = q$.  Denote by $N = \bigcap_{v\in M}N_{\widehat
    G}(v)$ the set of common neighbors of $M$ in ${\widehat G}$; by
  definition, $N_G(M) \subseteq N \subseteq N_{\widehat G}(M)$.  The
  theorem can be formulated as $N_{\widehat G}(M) = N$.  Suppose to
  its contrary, there exists a vertex $z\in N_{\widehat G}(M)\setminus
  N$, then we modify $\cal I$ into another interval set ${\cal I}' =
  \{I'_v|v\in V\}$.  We argue that the interval graph $\widehat G'$
  extracted from $\cal I'$ is a supergraph of $G$ and has strictly
  smaller size than $\widehat G$.
  This contradiction to the fact that $\widehat G$ is a minimum
  interval supergraph of $G$ will prove this theorem.
  In the sequel of this proof, the positive value $\epsilon$ is chosen
  to be small enough such that for the point $\ell$ under concern,
  neither $[\ell - \epsilon, \ell)$ nor $(\ell, \ell+\epsilon]$
  contains any endpoint of $\cal I$.  By assumption, $p\ne q$.

  {\em Case 1. $p > q$.}  Then $M$ induces a clique in ${\widehat G}$.
  We have $[q,p]\subset I_v$ for every $v\in M$, and $I_u\cap
  [q,p]\ne\emptyset$ for every $u\in N$.  We construct ${\cal I'}$ as
  follows.  For each $v\in M$, we set the endpoints of $I'_v$ to a
  distinct value in intervals $(q-\epsilon, q]$ and, respectively,
  $[p, p+\epsilon)$; in particular, we keep the right endpoint of
  $I'_x$ and the left endpoint of $I'_y$ to be $p$ and $q$,
  respectively.  For each $u\in V\setminus M$, we set $I'_u = I_u$.
  In the graph $\widehat G'$ represented by $\cal I'$, the subgraph
  induced by $M$ is a clique; the subgraph induced by $V\setminus M$
  is the same as $\widehat G - M$; and $M$ is completely connected to
  $N_G(M)\subseteq N$.  This verifies $G\subseteq \widehat G'$.
  On the other hand, for any $v\in M$, from $I'_v\subseteq I_v$ it can
  be inferred $N_{\widehat G'}(v) \subseteq N_{\widehat G}(v)$; it
  follows that $\widehat G'\subseteq \widehat G$.  By assumption,
  $I'_z$($ = I_z$) is either left to $q-\epsilon$ or right to
  $p+\epsilon$; hence $z\not\sim M$ in $\widehat G'$ and $\widehat
  G'\ne \widehat G$.  Putting them together we prove the case $p>q$.

  {\em Case 2. $p < q$.}  Then $x\not\sim y$.  We have $I_v\cap
  [p,q]\ne\emptyset$ for every $v\in M$, and $[p,q]\subset I_u$ for
  every $u\in N$.  (See Figure~\ref{fig:interval-representation}.)  We
  construct ${\cal I'}$ as follows.  Let $\ell$ be a point in $[p, q]$
  such that $K_{\ell}\setminus M$ has the minimum cardinality; without
  loss of generality, we may assume $\ell$ is different from any
  endpoint of intervals in $\cal I$.  For each $v\in M$, we set $I'_v$
  by projecting $I_v$ from $[p,q]$ to $[\ell-\epsilon,
  \ell+\epsilon]$.  For each $u\in V\setminus M$, we set $I'_u = I_u$.
  Let $\widehat G'$ be represented by $\cal I'$ and $N' = N_{\widehat
    G'}(M)$.  Observe that no interval in $\cal I$ has an endpoint in
  $[\ell-\epsilon, \ell+\epsilon]$; for each $u\in V\setminus M$, the
  interval $I'_u$ ($= I_u$) contains $\ell$ if and only if
  $[\ell-\epsilon, \ell+\epsilon] \subset I'_u$.  In other words, $N'
  = K_{\ell}\setminus M$.  For each $u\in N_G(M)$, the interval
  $I'_u$($=I_u$) contains $[p,q]$ which contains $[\ell-\epsilon,
  \ell+\epsilon]$ in turn; hence $N_G(M) \subseteq N'$.  The subgraphs
  induced by $M$ and $V\setminus M$ are the same as $\widehat G[M]$
  and $\widehat G - M$ respectively.  It follows that $G\subseteq
  \widehat G'$.

  It remains to show $||\widehat G'||< ||\widehat G||$, which is
  equivalent to $|E(\widehat G)'\cap (M\times V\setminus M)|<
  |E(\widehat G)\cap (M\times V\setminus M)|$, where $|E(\widehat
  G)'\cap (M\times V\setminus M)| = |M| |N'|$.  By the selection of
  $\ell$, for each point $\imath\in [p,q]$, we have
  $|K_{\imath}\setminus M| \ge |N'|$.  In other words, each vertex of
  $M$ has at least $|N'|$ neighbors in ${\widehat G}$, and thus
  $|E(\widehat G)\cap (M\times V\setminus M)| \ge |M| |N'|$, where
  equality only holds if $|N_{\widehat G}(v)\setminus M| = |N'|$ for
  every vertex $v\in M$.  We argue that this inequality has to be
  strict, which implies $||\widehat G'||< ||\widehat G||$.  As
  $z\not\in (N\cup M)$, it follows that $[p,q]\not\subset I'_z$ ($=
  I_z$) (see the thick/red edges in
  Figure~\ref{fig:interval-representation}).  If $p <
  \mathtt{right}(z) < q$ (see $z_1$ in
  Figure~\ref{fig:interval-representation}), then
  $|K_{\mathtt{right}(z)}\setminus M|> |K_{\mathtt{right}(z) +
    \epsilon}\setminus M|\ge |K_{\ell}\setminus M| = |N'|$.  As
  $\widehat G[M]$ is connected, there exists a vertex $v\in M$ such
  that ${\mathtt{right}(z)}\in I_v$, i.e., $|N_{\widehat
    G}(v)\setminus M| > |N'|$.  A symmetric argument applies if $p <
  \mathtt{left}(z) < q$.  Hence we may assume there exists no vertex
  $u\in V\setminus M$ such that $I'_u$ ($=I_u$) has an endpoint in
  $[p, q]$.  As a consequence, $K_{p-\epsilon}\setminus M =
  K_{p}\setminus M = N'$.  If $\mathtt{right}(z) < p$ (see $z_2$ in
  Figure~\ref{fig:interval-representation}), then there exists a
  vertex $v\in M$ adjacent to both $N'$ and $z$; e.g., the one with
  $\mathtt{left}(v)=\mathtt{left}(M)$ (not necessarily $x$).  A
  symmetric argument applies if $\mathtt{left}(z) > q$.  This
  completes the proof.
\end{proof} 

{Theorem}~\ref{thm:preserving-modules} will ensure preservation of any
connected module in perpetuity.  With the help of
{Lemma}~\ref{lem:partial-solution}, it can be further
strengthened to:
\begin{corollary}\label{lem:preserving-modules}
  Let $\widehat G$ be a minimum interval supergraph of graph $G$.  A
  connected module $M$ of any graph $\widehat G'$ satisfying $G
  \subseteq \widehat G' \subseteq \widehat G$ is a module of $\widehat
  G$.
\end{corollary}

The following theorem characterizes internal structures of modules,
and can be viewed as a complement to
Theorem~\ref{thm:preserving-modules}, which characterizes the external
behavior of modules in a minimum interval supergraph.
\begin{theorem}\label{thm:separable-modules}
  Let $\widehat G$ be a minimum interval supergraph of $G$ and $M$ be
  a {connected} module of $\widehat G$.  If $\widehat G[M]$ is not a
  clique, then for any minimum interval supergraph $\widehat G_M$ of
  $G[M]$, replacing $E(\widehat G[M])$ by $E(\widehat G_M)$ in
  $\widehat G$ gives a minimum interval supergraph $\widehat G'$ of
  $G$; and in particular, $\widehat G[M]$ is a minimum interval
  supergraph of $G[M]$.
\end{theorem}
\begin{proof}
  Let ${\cal I} = \{I_v|v\in V(G)\}$ be an interval representation of
  ${\widehat G}$.  We define $p = \min_{v\in M}\mathtt{right}(v)$ and
  $q = \max_{v\in M}\mathtt{left}(v)$; in particular
  $\mathtt{right}(x) = p$ and $\mathtt{left}(y) = q$.  As $\widehat
  G[M]$ is not a clique, $x\not\sim y$ and $p< q$.  It follows that
  $[p,q]\subset I_u$ for each $u\in N_{\widehat G}(M)$ and $I_u$ is
  disjoint from $[\mathtt{left}(M), \mathtt{right}(M)]$ for $u\not\in
  N_{\widehat G}[M]$.  It is easy to verify that $\widehat G'$
  corresponds to the following interval representation obtained by
  modifying $\cal I$: we build an interval representation for
  $\widehat G_M$ and project it to $[\mathtt{left}(M),
  \mathtt{right}(M)]$ to replace $\{I_v|v\in M\}$.  Observing that
  $\widehat G[M]$ is an interval supergraph of $G[M]$, if follows that
  $|\widehat G_M|\le|\widehat G[M]|$ and $|\widehat G'|\le||\widehat
  G||$.  Since $\widehat G$ is minimum, both inequalities have to be
  equalities; this completes the proof.
\end{proof}

According to {Corollary}~\ref{lem:preserving-modules}, any connected
module $M$ of an intermediate graph between $G$ and $\widehat G$ is a
connected module of $\widehat G$, to which
{Theorem}~\ref{thm:separable-modules} applies.  As a result, the
subgraph induced by $M$ in $\widehat G$ is either a clique, or a
minimum interval supergraph of $G[M]$.  In either case we have
\begin{corollary}\label{lem:separable-modules}
  Let $M$ be a connected module of graph $G$.  For any minimum
  interval supergraph $\widehat G_M$ of $G[M]$, there exists a minimum
  interval supergraph $\widehat G$ of $G$ such that $\widehat
  G_M\subseteq \widehat G[M]$.
\end{corollary}

\section{Shallow terminals in reduced graphs}\label{sec:shallow}
We say a graph is \emph{reduced} if it contains no hole or small
AW.\footnote{On this ostensibly counterintuitive notion a remark is
  worthwhile.  We reduce graphs by inserting edges; a reduced graph is
  thus a supergraph of the original graph.  We use ``reduced'' in the
  sense that its structure is simpler, and its size is closer to
  $\widehat G$ than $G$.}  As indicated by
Lemma~\ref{lem:fill-long-aw}, the shallow terminal shall be of special
interest during the disposal of a long AW.  The following
characterizations of shallow terminals 
were first proved on a class of less restricted graphs that excludes
small AWs and small holes in \cite{cao-12-interval-deletion}.  Reduced
graphs, excluding small AWs and \emph{all} holes, are a trivial subset
of them, whereby {Lemma}~\ref{lem:shallow-1} and \ref{lem:shallow-2}
apply to reduced graphs as well.  For completeness, their proofs are
repeated in Appendix.

\begin{lemma}\label{lem:shallow-1}
  Let $W$ be an AW with shallow terminal $s$ and base $B$ in a
  reduced graph, and $x$ is adjacent to $s$.
  \begin{enumerate}
  \item[(1)] Then $x$ is \tagged{full}{also} adjacent to the center(s)
    of $W$ (different from $x$).
  \item[(2)] Classifying $x$ with respect to its adjacency to $B$, we
    have the following categories:
    \begin{description}
    \item [(full)] $x$ is adjacent to every base vertex.\\
      Then $x$ is also adjacent to every vertex in $N(s)\setminus
      \{x\}$.
    \item [(partial)] $x$ is adjacent to some, but not all base vertices.\\
      Then there is an AW whose shallow terminal is $s$, one center is
      $x$, and base is a proper sub-path of $B$.
    \item [(none)] $x$ is adjacent to no base vertex.\\
      Then $x$ is adjacent to neither base terminals, and thus
      replacing the shallow terminal of $W$ by $x$ makes another AW.
    \end{description}
  \end{enumerate}
\end{lemma}
\begin{lemma}\label{lem:shallow-2}
  Let $W$ be an AW with shallow terminal $s$ and base $B$ in a reduced
  graph $G$.  Let $C=N(s)\cap N(B)$ and $M$ be the connected component
  of $G-C$ containing $s$. Then $C$ induces a clique and is completely
  connected to $M$.
\end{lemma}

Let $ST(G)$ denote the set of shallow terminals of a reduced graph
$G$.  The lemmas above indicate a nice structure for $ST(G)$.
Observe that for a module $M$ in a chordal graph, at least one of $M$
and $N(M)$ induces a clique.  We say $M$ is \emph{simplicial module} if it
is a connected module and $N(M)$ induce a clique.  We remark that this
name is suggested by {Lemma}~\ref{lem:only-terminals-are-simplicial}
and the fact that any $s\in M$ is simplicial in $G - (M\setminus
\{s\})$.
\begin{lemma}\label{lem:simple-module}
  Let $M$ be a simplicial module of a reduced graph $G$.  If an AW $W$
  contains a vertex $x\in M$ and $W\not\subseteq M$, then $x$ is a
  terminal of $W$.  Moreover, if $G$ is connected, then $G - M$ is
  connected.
\end{lemma}
\begin{proof}
  By {Lemma}~\ref{lem:module-exchangeable}, $x$ is the only vertex in
  $M\cap W$, thus simplicial; the first assertion follows from
  {Lemma}~\ref{lem:only-terminals-are-simplicial}.  As $N(M)$ induces
  a clique, if a pair of vertices $u,v\not\in M$ is connected by a
  path intersecting $M$, then there is a $u$-$v$ path avoiding $M$;
  the second assertion follows.
\end{proof}

\begin{theorem}\label{thm:shallow-is-module}
  In the subgraph induced by $ST(G)$, each connected component makes a
  simplicial module of $G$.
\end{theorem}
\begin{proof}
  Given any AW, we can use Lemma~\ref{lem:shallow-2} to construct the
  pair of sets $M$ and $C$; using definition we can verify that $M$
  makes a simplicial module of $G$.  If $C \cap ST(G) = \emptyset$ then
  $M$ is a connected component of the subgraph induced by $ST(G)$.
  Hence we assume otherwise, and let $W$ be another AW with shallow
  terminal $s\in C \cap ST(G)$.  As $s$ is adjacent to every other
  vertex of $C\cup M$, the base $B$ of $W$ is disjoint from $C\cup M$;
  i.e., $M$ is adjacent to $s$ but not $B$.  The new set $M'$ obtained
  by applying Lemma~\ref{lem:shallow-2} on $W$ contains both $M$ and
  $s$; it is also a simplicial module of $G$.  This process can be
  repeated for a finite number of steps, until a connected component
  of the subgraph induced by $ST(G)$, also a simplicial module, is found.
\end{proof}

A reduced graph $G$ contains no hole or small AW; $ST(G)$ intersects
every long AW, and thus $G - ST(G)$ is an interval graph.  On the
other hand, applying {Lemma}~\ref{lem:simple-module} repetitively on
connected components of $ST(G)$ gives:

\begin{corollary}\label{lem:connected-backbone}
  If a reduced graph $G$ is connected, then $G-ST(G)$ is a connected
  interval graph.
\end{corollary}

To apply the results of this section we need to first
find $ST(G)$.  We check for each triple of vertices whether they form
an AT or not, and identify an AW for them if yes.  The
AW is necessarily a long AW and contains a shallow terminal.  Clearly
it takes polynomial time to check all triples.  The following lemma
assures us that all shallow terminals can be found as such.
\begin{lemma}\label{lem:shallow-is-shallow}
  In a reduced graph, all AWs with the same set of terminals have the
  same shallow terminal.
\end{lemma}
\begin{proof}
  Let $(s: c_1,c_2: l,h; t,r)$ be the frame of an AW $W$.  We consider
  the distance between $l$ and $r$ in $G - N[s]$, which cannot be $1$
  by definition.  As a shallow terminal is in distance either $2$ or
  $3$ to a base terminal, if the distance between $l$ and $r$ in $G -
  N[s]$ is strictly larger than $3$, then this assertion must hold
  true.  Suppose it is $2$ and $x\not\in N[s]$ is a common neighbor of
  $l$ and $r$; clearly $x\not\in W$.  As there cannot be a hole by
  assumption, $x$ must be adjacent to both $h$ and $t$.  Noting
  $x\not\sim s$, this contradicts
  Lemma~\ref{lem:common-neighbor-of-base}.
  Suppose now it is $3$ and $(l x y r)$ is a shortest $l$-$r$ path in
  $G - N[s]$; then $l\not\sim y$ and $x\not\sim r$.  It might happen
  that $x$ or $y$ is in $W$; in particular, $x=h$ or $y=t$, but not
  both.  If $x,y\not\in W$, then $x,y$ make a cycle with the path $(l
  h c_1 c_2 t r)$.  From the nonexistence of holes and the known
  adjacencies, it can be inferred that $x\sim h$ and $y\sim t$; and at
  least one of $x\sim t$ and $y\sim h$ holds true.  If $x=h$ or $y=t$,
  then the other vertex is adjacent to both $h$ and $t$.  Therefore,
  we always ends with a vertex in $N(h)\cap N(t)\setminus N(s)$,
  contradicting Lemma~\ref{lem:common-neighbor-of-base}.  This
  completes this proof.
\end{proof}

It should be noted that this does not rule out the possibility of the
shallow terminal of an AW being a base terminal of another AW; if this
happens, these AWs necessarily have at least one different terminal.
Indeed, for any AW not fully contained in $ST(G)$ in a reduced graph
$G$, we can conclude from {Theorem}~\ref{thm:shallow-is-module} and
Lemma~\ref{lem:simple-module} that
\begin{inparaitem}
\item its shallow terminal is in $ST(G)$;
\item its base terminals might or might not be in $ST(G)$; and
\item all other vertices are disjoint from $ST(G)$.
\end{inparaitem} 

Finally, our branching shall be conducted on a ``locally minimal'' AW
which can be found as follows.
\begin{lemma}\label{lem:min-at}
  For any $s\in ST(G)$ in a reduced graph $G$, there is an AW whose
  shallow terminals is $s$ and whose base is completely connected to
  $N(s)\setminus ST(G)$.  Moreover, such an AW can be found in
  polynomial time.
\end{lemma}
\begin{proof}
  We start from any AW $W$ with shallow terminal $s$.  We use
  Lemma~\ref{lem:shallow-1} to categorize vertices in $N(s)\setminus
  ST(G)$ with respect to $W$.  None of them cannot be in category
  ``none,'' as otherwise such a vertex is a shallow terminal and has
  to be in $ST(G)$.  If every vertex in $N(s)\setminus ST(G)$ is in
  category ``full,'' then $W$ is already what we need and we are done.
  Hence let us assume $x$ is in category ``partial,'' then we have
  another AW with shallow terminal $s$ and a strictly shorter base.
  Applying this argument repeatedly will eventually procure an AW with
  shallow terminal $s$ such that every vertex in $N(s)\setminus ST(G)$
  is in category ``full.''  It is easy to verify that this procedure
  can be implemented in polynomial time; this completes this proof.
\end{proof}

\section{The algorithm}
\label{sec:alg}

Now we are ready to present the main algorithm and prove
Theorem~\ref{thm:main-alg}.  Our basic strategy is a sandwich
approach, which either inserts edges to $G$, or excludes some other
edges by setting them as ``avoidable.''  As such we narrow the search
space from both sides, until the objective graph is obtained.

The execution of the algorithm, an intermixed application of several
branching rules, can be described as a search tree of which every node
contains a pair ($G,A$), where $A$ denotes the set of ``avoidable''
edges.  On a non-leaf node of this tree, we make mutually exclusive
decisions, each generating a different child node.
Let ($G,A$) and ($G',A'$) be contained in a parent node and,
respectively, a child node in the search tree; it holds $G\subseteq
G'$ and $A\subseteq A'$.  We say an interval supergraph $\widehat G$
of $G$ is \emph{feasible} for ($G,A$) if $E(\widehat G)$ is disjoint
from $A$, even its size is larger than $||G||+k$.  Any graph feasible
for ($G',A'$) is also feasible for ($G,A$).

We say an input instance ($G,k$) of \textsc{interval completion} is a
``YES'' instance if the size of a minimum interval supergraph of $G$
is at most $||G|| + k$; a ``NO'' instance otherwise.  We define
$\mathtt{mis}$ to be the size of minimum interval supergraphs of the
input graph for a ``YES'' instance; and $||G|| +k$ for a ``NO''
instance.  By definition, it always holds $\mathtt{mis}\le ||G|| + k$.

The way we prove the correctness of our algorithm is by showing if a
node in the search tree has a feasible supergraph of size upper
bounded by $\mathtt{mis}$, then at least one of its children nodes has
a feasible supergraph of the same size.  For a ``NO'' instance, this
holds trivially: no matter which edges are chosen to be inserted or
forbidden in any step, the monotonicity of $G$ and $A$ ensures every
path in the search tree faithfully ends with ``NO.''  Hence we can
focus on ``YES'' instances, where the root node surely satisfies this
condition.  With inductive reasoning, we conclude that there is a
leaf node containing an interval graph $\widehat G$ of size
$\mathtt{mis}$.  To such a leaf node there is a unique path from the
root, and $\widehat G$ is feasible for every node in the path.

On the complexity analysis, we focus on the number of leaves of the
search tree the algorithm traverses.  It is achieved by bounding the
number of children nodes of a non-leaf node with respect to the
decrease of measure.  As we have no intention of optimizing the order
of the polynomial factor, for the generation of a child node, we are
satisfied if it can be executed in polynomial time.

To facilitate the recursive calls and inductive proofs, we augment the
algorithm inputs and strengthen inductive hypothesis as follows.  In
addition to the graph $G$ and parameter $k$, our algorithm takes as
inputs:
\begin{itemize}
    \itemsep0mm   
  \item[$U$:] a set  of shallow terminals; 
  \item[$\cal X$:] a set  of unchangeable frames; and
  \item[$A$:] a set  of ``avoidable'' edges.
\end{itemize}
They are related as follows.  Let $M$ be a connected component of
$G[U]$.  Within $\cal X$ there is an unchangeable frame ($s: c_1,c_2:
l,h; t,r$), denoted by $X(M)$, such that $s$ is the only vertex in
$X(M)\cap M$, and $\{l c_2, c_1 r, h t, x h, x t| x\in M\}\subseteq A$
(Lemma~\ref{lem:fill-long-aw}).  We will use $h(M)$, $t(M)$, and
$s(M)$ to refer to the vertices $h,t$, and $s$ respectively in $X(M)$.
The set $A$ is disjoint from $E(G)$, and both of them increase only;
and $||G|| + k$ remains the same throughout.  The original instance
($G,k$) is supplemented with empty $U$, $\cal X$, and $A$; it thus
makes ($G, k, \emptyset, \emptyset, \emptyset$).  We define the
measure as $m = k - |U|$.

Aside from the aforementioned condition on size, we delineate $6$
other conditions that hold throughout.  All the $7$ invariants are
summarized in Figure~\ref{fig:invariants}.  It should be easy to
verify the base case, ($G, k, \emptyset, \emptyset, \emptyset$),
satisfies all the conditions.
\begin{figure}[ht]
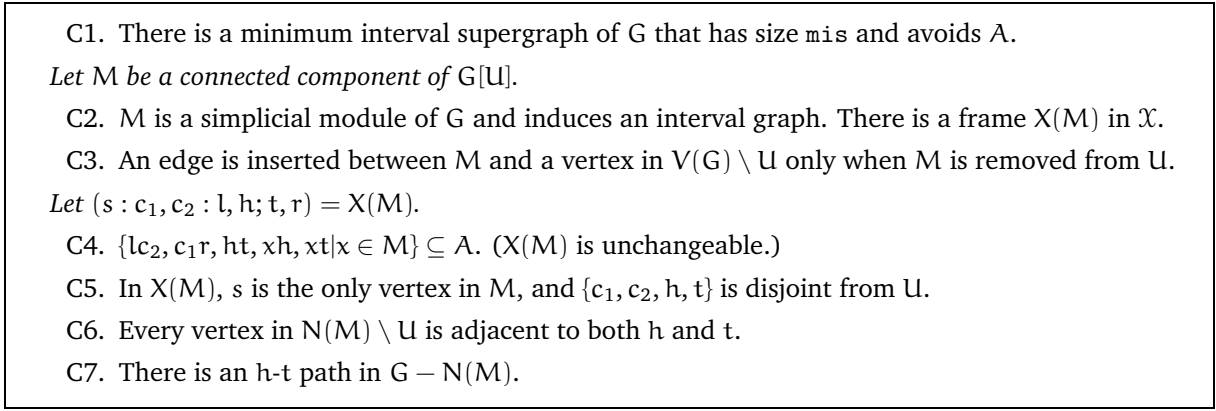

  \centering
\fbox{\parbox{0.95\linewidth}{
    \vspace*{-2mm} 
  \begin{enumerate}[ C1.]
    \itemindent6mm\itemsep0mm 

  \item\label{c:solution} There is a minimum interval supergraph of 
    $G$ that has size $\mathtt{mis}$ and avoids $A$.
  \item[] \hspace*{-9mm}{\em Let $M$ be a connected component of $G[U]$.}
  \item\label{c:module}  $M$ is a simplicial module of $G$ and induces an 
    interval graph.  There is a frame $X(M)$ in $\cal X$.
  \item\label{c:exit-u} An edge is inserted between $M$ and a vertex in 
    $V(G)\setminus U$ only when $M$ is removed from $U$.
  \item[] \hspace*{-9mm}{\em Let $(s: c_1, c_2:  l,h; t,r) = X(M)$.}
  \item\label{c:A} $\{l c_2, c_1 r, h t, x h, x t| x\in M\}\subseteq A$.  
    ($X(M)$ is unchangeable.)
  \item\label{c:frame}  In $X(M)$, $s$ is the only vertex in $M$, and $\{c_1, 
    c_2,h,t\}$ is disjoint from $U$.  
  \item\label{c:center} Every vertex in $N(M)\setminus U$ is adjacent to both
    $h$ and $t$.
  \item\label{c:base} There is an $h$-$t$ path in $G - N(M)$.
  \end{enumerate}
    \vspace*{-2mm}
}}
  \caption{Invariants during our algorithm}
  \label{fig:invariants}
\end{figure}

The remainder of this section is devoted to presenting the algorithm
and proving the following lemma.  

\begin{lemma}
  On input ($G, k, U, {\cal X}, A$) that satisfies all conditions
  C1-7, the algorithm runs in $6^{k - |U|}\cdot n^{{\cal O}(1)}$.
  Moreover, at the exit of this algorithm,
  \begin{itemize}
  \item[($\bot$)]: The algorithm returns a minimum interval supergraph
    of $G$, and $U=\emptyset$.
  \end{itemize}
\end{lemma}

An immediate implication of C1-7 is the following termination
condition, which enables us to prune many subtrees.
\begin{claim}
  If $k<|U|$ holds in a node of the search tree, then any feasible
  solution to ($G,A$) has size strictly larger than $||G|| + k$.
\end{claim}
\begin{proof}
  By Lemma~\ref{lem:fill-long-aw} and C\ref{c:base}, for every $s\in
  U$, we need to insert at least one edge between $s$ and
  $V(G)\setminus U$.  \renewcommand{\qedsymbol}{$\lrcorner$}
\end{proof}

The algorithm consists of two phases.  Phase~I partitions the graph
into two disjoint interval subgraphs, while Phase~II merges them.
Phase~I iteratively executes two procedures, until the required
condition is achieved.  Phase~II runs one single procedure and only
once.  We now describe each procedure, analyze its runtime, and verify
the inductive hypothesis.  (See Figure~\ref{fig:alg-completion} for an
outline of the algorithm)

\subsection{Phase~I}
The aim of this phase is to partition the graph into two disjoint
interval subgraphs $G[U]$ and $G - U$.  As $U$ always induces an
interval graph throughout (C\ref{c:module}), the focus shall be laid
on $G - U$.  Procedure 1 reduces $G - U$ by breaking all its holes and
small AWs.  Procedure 2 takes care of long AWs; since it works only on
reduced graphs, but the disposal of one long AW might introduce holes
and/or small AWs, between the disposal of two long AWs we need to
rerun procedure 1.

\paragraph{Procedure 1. Reducing $G - U$.}
This procedure repeatedly finds a small AW or a hole, and uses
Lemmas~\ref{lem:holes} or \ref{lem:6-enough}, respectively, to fill it.
On a hole $H$, we branch on inserting one of the at most $4^{|H|-3}$
minimal sets specified in Lemma~\ref{lem:holes}; in each branch, the
measure decreases by $|H| - 3$.  On a small AW, we branch on inserting
one of the $6$ edges specified in Lemma~\ref{lem:6-enough}; in each
branch, the measure decreases by $1$.  

C\ref{c:solution} is ensured by Lemma~\ref{lem:holes} and
\ref{lem:6-enough}.  Since all operations are conducted in $G - U$,
C2-7 are satisfied.

At the end of this procedure, if $G - U$ is already an interval graph,
then we are done with Phase~I and turn to Procedure~3; otherwise we go
to Procedure~2.
\paragraph{Procedure 2.  Coping with long AWs in $G - U$.}
The only entry to Procedure~2 is from Procedure~1; hence $G - U$ is a
reduced but non-interval graph.  Then $ST(G - U)$ is nonempty, and can
be computed with Lemma~\ref{lem:shallow-is-shallow}.  We take a
connected component $M'$ of the subgraph induced by $ST(G - U)$, and
let $\{M_1,\dots,M_p\}$ be all connected components in $G[U]$ that are
adjacent to $M'$.  Here $p=0$ if $M'$ is not adjacent to $U$.  

{\em Case 1.} If there is $1\le i\le p$ and a vertex $x$ such that $x$
is adjacent to both $h(M_i)$ and $t(M_i)$, but nonadjacent to
$s(M_i)$, then we insert edges $x\times M_i$, remove $M_i$ from $U$,
and remove $X(M_i)$ from $\cal X$.  Observe that $h,t\not\in U$
(C\ref{c:frame}), and a vertex in $U$ cannot be a common neighbor of
$h$ and $t$ (C\ref{c:module}).  Hence $x\not\in U$.

This step generates only one child; since both $k$ and $|U|$ decrease
by $|M_i|$, the measure remains unchanged.  This step can be checked
in polynomial time.  Moreover, noting that it decreases $k$, this step
can be executed at most $k$ times during the whole algorithm.

C\ref{c:solution} is ensured by
Lemma~\ref{lem:common-neighbor-of-base} (C\ref{c:A}).  As no vertex is
put into $U$, and $M_i$ is removed from $U$ while no edge incident to
$U\setminus M_i$ is inserted, C2-7 are verified.

{\em Case 2.} If the pair $M_i$ and $x$ above is found, then we end
this procedure.  Otherwise we proceed as follows.  Let $M = M' \cup
\bigcup^p_{i=1} M_i$.  Since for each $1\le i\le p$, the set $M_i$ is
a module of $G$ (C\ref{c:module}), $M_i\sim M'$ means every vertex in
$M_i$ is adjacent to $M'$.  In other words, $\bigcup^p_{i=1} M_i =
N(M') \cap U = M\cap U$.  The following proposition characterizes $M$.

\begin{claim}\label{claim-20}
  The set $M$ is a simplicial module of $G$.  For each $1\le i\le p$, it
  contains a long AW whose frame is $X(M_i)$.
\end{claim}
\begin{proof}
  According to Theorem~\ref{thm:shallow-is-module}, $M'$ is a
  simplicial module of $G - U$.  If $M'$ is not adjacent to $U$, then
  $M = M'$.  Using definition we can verify that $M$ is also a
  simplicial module of $G$, and the statement holds vacuously.  Hence
  we may assume otherwise.

  There is a frame $X(M_1)$ in $\cal X$ (C\ref{c:module}); let it be
  $(s:c_1,c_2: l,h;t,r)$ where $s\in M_1$.  By assumption, $M_1$ is
  adjacent to some vertex $v\in {M'}$, which is then adjacent to both
  $h$ and $t$ (C\ref{c:center}).  Noting that $h,t\not\in U$
  (C\ref{c:frame}), we must have $h,t\in N_{G-U}[M']$.  As $M'$ is a
  simplicial module of $G - U$, a vertex in $N_{G-U}(M')$ is adjacent
  to every other vertex in $N_{G-U}[M']$.  From the nonadjacency of
  $h$ and $t$, we can conclude $h,t\in M'$, and they are completely
  connected to $N_{G-U}(M')$.  Then $N_{G-U}(M')$ is completely
  connected to $s$, and also $M_1$ as $M_1$ is a module of $G$.

  On the other direction, $N(M_1)$ induces a clique (C\ref{c:module}).
  Consider any neighbor $v$ of $M_1$ in $V(G)\setminus M'$.  It is
  adjacent to $M'$ and not contained in $U$; as $M'$ is a module of
  $G-U$, it follows that $v$ is completely connected to $M'$.

  Arguments above also apply to $M_i$ for $2\le i\le p$.  As a result,
  if a vertex $x\in V(G)\setminus M$ is adjacent to $M'$, then it is
  adjacent to every vertex in $N(M')\cap U$, and vice versa.  This
  verifies $M$ is a module of $G$.  By definition, $M$ is connected,
  and $N(M) = N_{G - U}(M')$ as it is disjoint from $U$; therefore,
  $M$ is a simplicial module of $G$.

  We have already shown $h$ and $t$ are in $M'$.  We now consider
  other vertices of $X(M_1)$, i.e., $\{c_1, c_2, l, r\}$.  As $l$ and
  $c_2$ are nonadjacent and are both adjacent to some vertex of $M'$,
  which is a subset of $M$, they have to be in $M$.  A symmetric
  argument will imply $c_1, r\in M$.  Thus, $M$ contains every vertex
  of $X(M_1)$.  Finally let $P$ be the shortest $h$-$t$ path in $G -
  N[M_1]$ (C\ref{c:base}).  As $G - U$ is chordal, every inner vertex in $P$ is
  adjacent to both $c_1$ and $c_2$; hence also in $M'$. This completes
  the proof.
  \renewcommand{\qedsymbol}{$\lrcorner$}
\end{proof}

If the module $M$ does not induce an interval graph, then we make a
recursive call to fill it.  Specifically, we invoke our algorithm with
($G[M], k, U\cap M, \{X(M_i)|1\le i\le p\}, A\cap M^2$).  The
following claim captures the validity of this invocation:
\begin{claim}
  The tuple ($G[M], k, U\cap M, \{X(M_i)|1\le i\le p\}, A\cap M^2$)
  satisfies C1-7, where $\mathtt{mis}$ is set to be the size of
  minimum interval supergraphs of $G[M]$.
\end{claim}
\begin{proof}
  If $A\cap M^2 = \emptyset$, then C\ref{c:solution} holds vacuously.
  Otherwise, let $\widehat G$ be a minimum interval supergraph of $G$
  that avoids $A$ (C\ref{c:solution}).  The subgraph $\widehat G[M]$
  is then not complete; according to
  Theorem~\ref{thm:separable-modules}, $\widehat G[M]$ is a minimum
  interval supergraph of $G[M]$ and avoids $A\cap M^2$.  C2-7 follow
  from {Claim}~\ref{claim-20}.  \renewcommand{\qedsymbol}{$\lrcorner$}
\end{proof}

With an inductive reasoning, we may assume this invocation returns a
minimum interval supergraph of $G[M]$; let it be $\widehat G_M$.
According to Corollary~\ref{lem:separable-modules}, this is correct.
Also note that at the return $U$ is consequently disjoint from $M$
($\bot$).  We now analyze the runtime of this call.  Let $k_M$ be the
number of edges inserted during it, i.e., $k_M = ||\widehat G_M|| -
||G[M]||$, and $U_M = U\cap M$.  The measure decreases by $k_M -
|U_M|$, while this invocation takes $6^{k_M - |U_M|}$.  

In this juncture, the configuration becomes ($G', k' = k-k_M, U' =
U\setminus M, {\cal X}, A$), where $G'$ differs from $G$ only in edges
$M^2$, i.e., $G'[M] = \widehat G_M$.  It is easy to verify that C1-7
remain true.

By definition and Claim~\ref{claim-20}, $M$ is a set of shallow
terminals of $G - (U\setminus M)$, and it remains a set of shallow
terminals of $G' - U'$.  We can take any vertex $s\in M$, and use
Lemma~\ref{lem:min-at} to find a long AW with frame $\{s: c_1, c_2:
l,h;t,r\}$.  We branch into $6$ direction as follows:
\begin{itemize}
  \itemsep0mm
\item insert one of the $3$ edges $\{l c_2, c_1 r, h t\}$ and
  decrease $k$ by $1$;
\item insert either $h\times M$ or $t\times M$ and decrease $k$ by
  $|M|$; or
\item add $M$ into $U$, $X(M) = \{s: c_1, c_2: l, h; t, r\}$ into
  $\cal X$, and $\{l c_2, c_1 r, h t, x h, x t| x\in M\}$ to $A$.
\end{itemize}
In each of the $6$ directions, the measure decreased by at least $1$:
either $k$ decreases, or $|U|$ increases.

\begin{claim}
  In at least one branch C1-7 remain true.
\end{claim}
\begin{proof}
  Let $\widehat G$ be a minimum interval supergraph of $G$ that avoids
  $A$ (C\ref{c:solution}).  If $\widehat G$ contains one of the edges
  $\{l c_2, c_1 r, h t\}$, or one of the sets $h\times M$ and
  $t\times M$, then C\ref{c:solution} remains true at this branch.
  The edge(s) are inserted in $G - U$ only, and hence C2-7 remain true.

  Hence we may assume $\widehat G$ contains none of the specified
  edges, which implies C\ref{c:solution}.  In this direction, $M$ is
  newly inserted to $U$ and they are nonadjacent.  By
  Claim~\ref{claim-20}, C2-7 are satisfied on $M$.  While no other
  vertices in $U$ are impacted; hence C2-7 remain true on them.
  \renewcommand{\qedsymbol}{$\lrcorner$}
\end{proof}

At the end of this procedure, if $G - U$ is already an interval graph,
then we are done with Phase~I and turn to Procedure~3 directly;
otherwise we come back to Procedure~1.

\subsection{Phase~II}
We are now at the second phase, where, a priori, both $G[U]$ and $G -
U$ are interval graphs.
\paragraph{Procedure 3.  Merging $U$ to $G - U$.}
We construct an interval representation $\cal I$ for $G - U$.  For
each connected component $M$ of $G[U]$, both $h(M)$ and $t(M)$ are in
$V(G)\setminus U$ (C\ref{c:frame}) and are nonadjacent.  Assume
without loss of generality, $I_{h(M)}$ goes left to $I_{t(M)}$; we
define $p_M = \mathtt{right}(h(M))$ and $q = \mathtt{left}(t(M))$.
Let $\ell$ be the point in $(p, q)$ that minimizes $|K_{\ell}|$ among
them satisfying that $K_{\ell}\times M$ is disjoint from $A$.  We
insert edges to completely connect $M$ and $K_{\ell}\setminus N(M)$,
and remove $M$ from $U$.  We stop at $U=\emptyset$.  If the total
number of edges inserted is larger than $k$, then we return ``NO'';
otherwise the interval graph obtained.
\begin{claim}\label{claim-40}
  The graph $\widehat G$ obtained as above is a minimum interval
  supergraph of $G$ that avoids $A$.
\end{claim}
\begin{proof}
  To show it makes an interval graph, it suffices to build an interval
  representation as follows.  Without loss of generality, we may
  assume $\ell$'s selected for different connected components are
  different and avoid any endpoint of $\cal I$.  We build an interval
  representation for $G[M]$, and project it to
  $[\ell-\epsilon,\ell+\epsilon]$.  It is easy to verity this interval
  representation corresponds to $\widehat G$.

  To show it is minimum, we show the edges inserted to each connected
  component $M$ of $G[U]$ is minimum.  According to
  Corollary~\ref{lem:fill-long-aw-2} and
  Lemma~\ref{thm:interval-representation}, every vertex in $M$ is in
  some minimal $h(M)$-$t(M)$ separator $S$ in any interval supergraph
  of $G$.  This separator has to be a clique, and contain at least a
  minimal $h$-$t$ separator in the subgraph $G - U$.  Therefore, we
  need to find some $h$-$t$ separator $S'$ in the subgraph $G - U$,
  and completely connect it to $M$ (C~\ref{c:module}).  By the
  selection of $\ell$, we need to insert at least $M\times
  K_{\ell}\setminus N(M)$ edges to $M$.  This verifies that $\widehat
  G$ is minimum and completes the proof.
  \renewcommand{\qedsymbol}{$\lrcorner$}
\end{proof}

This procedure runs in polynomial time.  For each connected component
$M$ in $G[U]$, at least $|M|$ edges in $M\times V(G)\setminus U$ are
inserted (C~\ref{c:base}).  Therefore, the measure is non-increasing.

C\ref{c:solution} is ensured by {Claim}~\ref{claim-40}.  Since
$U=\emptyset$ after this step, C2-7 hold vacuously.  Exit condition
$\bot$ is also satisfied.

\section{Concluding remarks}
\label{sec:remark}
Theorem~\ref{thm:shallow-is-module} only holds in graphs free of
$3$-nets and $3$-tents (see Lemma~\ref{lem:common-neighbor-of-base}
and especially Figure~\ref{fig:common-neighbor}).  Hence in the
reduction step, we do away with them and make sure $d>3$ in the
remaining graph.  Interestingly, they are also the largest $\dag$- and
$\ddag$-AWs, respectively, that admit a $6$-way branching (see
Lemma~\ref{lem:6-enough} and especially Figure~\ref{fig:fill-edges}).
This indicates we might have reached the limit of basic bounded
search.  And to further lower the exponential factor in the time
complexity, new observations and approach are required.  We leave it
open for the existence of a sub-exponential algorithm and more
rivetingly, a polynomial kernel.


We present the algorithm with the bare essentials of modules.  One may
insert one more preprocessing step to our algorithm.  That is, we may
compute a modular decomposition for the graph, and then on insertion
of an edge between a module and others, we fill in a all-or-none
manner.  It might speed up the algorithm on graphs with many
nontrivial modules.  We leave this for later work on algorithmic
engineering.  As shown in Figure~\ref{fig:broken-modules}, there are
interval supergraphs that do break disconnected modules, so the
connected condition in Theorem~\ref{thm:preserving-modules} is
essential.  We remark that an alternative way is to replace ``for
any'' by ``there exists,'' as in the following statement.
\begin{lemma}
  For any module $M$ of graph $G$, there exists a minimum interval
  supergraph $\widehat G$ of $G$ such that $M$ is a module of
  $\widehat G$.
\end{lemma}

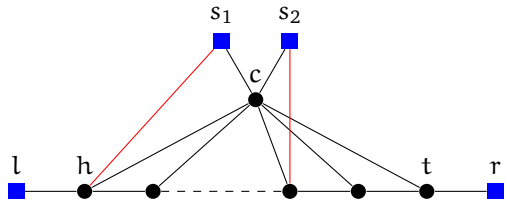
\begin{figure*}[ht]
  \centering
    \begin{tikzpicture}[scale=.45]
    \node [corner,label=above:$s_1$] (s1) at (-1,4.44) {};
    \node [corner,label=above:$s_2$] (s2) at (1,4.44) {};
    \node [corner,label=above:$l$] (a) at (-7, 0) {};
    \node [special,label=above:$h$] (a1) at (-5, 0) {};
    \node [vertex] (a2) at (-3, 0) {};
    \node [vertex] (bi) at (1, 0) {};
    \node [vertex] (b2) at (3, 0) {};
    \node [special,label=above:$t$] (b1) at (5, 0) {};
    \node [corner,label=above:$r$] (b) at (7, 0) {};
    \node [special,label=above:$c$] (c) at (0,2.7) {};

    \draw[] (a) -- (a1) -- (a2) (c) -- (b1);;
    \draw[] (b) -- (b1) -- (b2) -- (bi) -- (c) -- (b2);
    \draw[] (a1) -- (c) -- (a2) (s1) -- (c) -- (s2);
    \draw[dashed] (a2) -- (bi);
    \draw[red] (s1) -- (a1) (s2) -- (bi);
    \end{tikzpicture}
    \caption{A non-connected module $\{s_1, s_2\}$.}
    \label{fig:broken-modules}
\end{figure*}

In the intermediate step of our algorithm, we have some edges
forbidden.  One should not confuse this with the \textsc{interval
  sandwich} problem \cite{golumbic-93-temporal-reasoning,
  golumbic-95-graph-sandwich} (see also
\cite{fomin-12-subexponential-fill-in}).  The latter generalizes
\textsc{interval supergraph} by imposing an \emph{arbitrary} set $F$
of edges that are not allowed to be inserted.  The crucial difference
is that a minimum solution to an instance $(G,F)$ of \textsc{interval
  sandwich} is not necessarily a minimum supergraph of $G$.  This
explains why we use ``avoidable'' instead ``forbidden'' for our
algorithm.  The new challenge is surely that modules are not
necessarily preserved, and our algorithm will cease to work.  A
natural question is, can we adapt our algorithm to work on
\textsc{interval sandwich}?

\paragraph{Acknowledgment.}  I am grateful to Sylvain Guillemot for
his careful reading of an early version of this manuscript and helpful
comments.

\appendix
\newpage
\section{Outline of the main algorithm}
\vspace{-6mm}
\begin{figure}[ht]
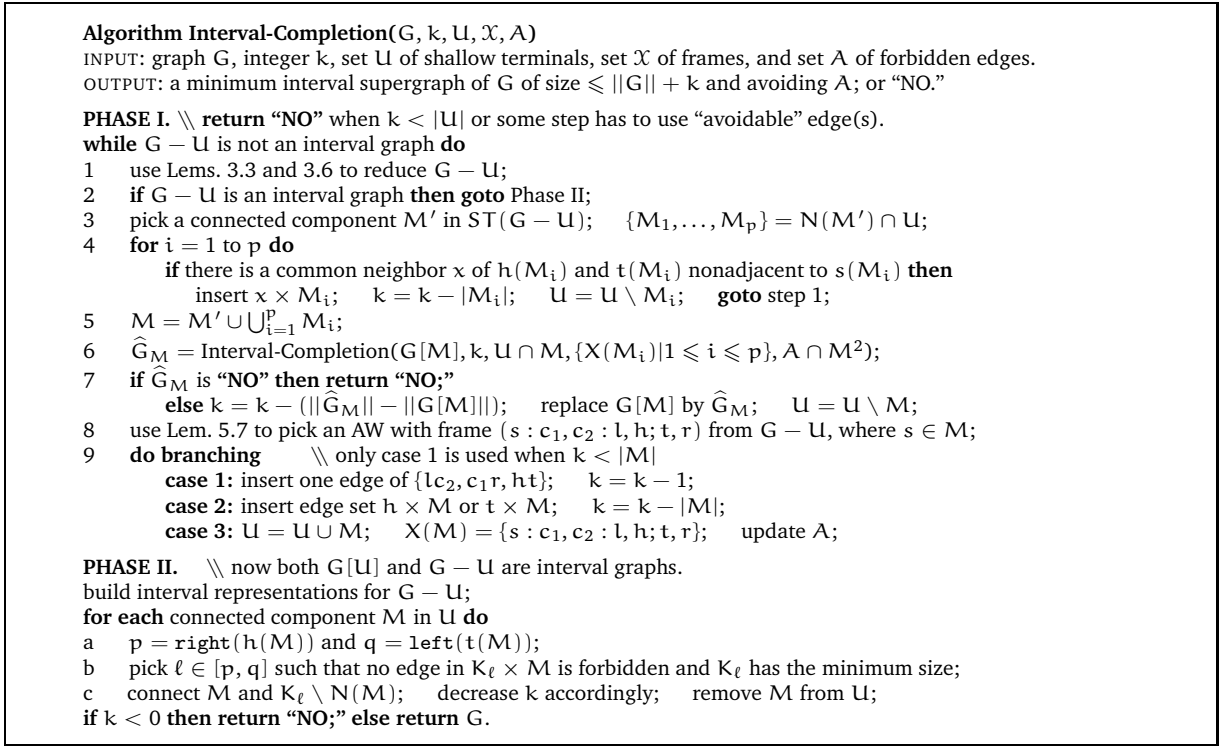

\setbox4=\vbox{\hsize28pc \noindent\strut
\begin{quote}
  \vspace*{-5mm} \footnotesize

  {\bf Algorithm Interval-Completion($G,k,U,\mathcal{X}, A$)}
  \\
  {\sc input}: graph $G$, integer $k$, set $U$ of shallow terminals,
  set $\cal X$ of frames, and set $A$ of forbidden edges.
  \\
  {\sc output}: a minimum interval supergraph of $G$ of size $\leq
  ||G||+k$ and avoiding $A$; or ``NO.''

  {\bf PHASE I.}  $\setminus\!\!\setminus$ {\bf return ``NO''} when
  $k<|U|$ {or} some step has to use ``avoidable'' edge(s).
  \\
  {\bf while} $G - U$ is not an interval graph {\bf do}
  \\
  1 \hspace*{3mm} use Lems.~\ref{lem:holes} and \ref{lem:6-enough} to
  reduce $G - U$;
  \\
  2 \hspace*{3mm} {\bf if} $G - U$ is an interval graph {\bf then
    goto} Phase~II;
  \\
  3 \hspace*{3mm} pick a connected component $M'$ in $ST(G - U)$;
  $\quad$ $\{M_1,\dots,M_p\} = N(M')\cap U$;
  \\
  4 \hspace*{3mm} {\bf for} $i=1$ to $p$ {\bf do}
  \\
  \hspace*{10mm} {\bf if} there is a common neighbor $x$ of $h(M_i)$
  and $t(M_i)$ nonadjacent to $s(M_i)$ {\bf then}
  \\
  \hspace*{14mm} insert $x\times M_i$; $\quad$ $k = k - |M_i|$;
  $\quad$ $U = U\setminus M_i$; $\quad$ {\bf goto} step 1;
  \\
  5 \hspace*{3mm} $M = M'\cup \bigcup^p_{i=1}M_i$;
  \\
  6 \hspace*{3mm} $\widehat G_M=$ Interval-Completion($G[M], k, U\cap
  M, \{X(M_i)|1\le i\le p\}, A\cap M^2$);
  \\
  7 \hspace*{3mm} {\bf if} $\widehat G_M$ is {\bf ``NO'' then return
    ``NO;''}
  \\
  \hspace*{10mm} {\bf else} $k = k - (||\widehat G_M|| - ||G[M]||)$;
  $\quad$ replace $G[M]$ by $\widehat G_M$; $\quad$ $U = U\setminus
  M$;
  \\
  8 \hspace*{3mm} use Lem.~\ref{lem:min-at} to pick an AW with frame
  $(s:c_1,c_2:l,h;t,r)$ from $G - U$, where $s\in M$;
  \\
  9 \hspace*{3mm} {\bf do branching} $\qquad$$\setminus\!\!\setminus$
  only case 1 is used when $k< |M|$
  \\
  \hspace*{10mm} {\bf case 1:} insert one edge of $\{l c_2, c_1 r, h
  t\}$; $\quad$ $k = k - 1$;
  \\
  \hspace*{10mm} {\bf case 2:} insert edge set $h\times M$ or
  $t\times M$; $\quad$ $k = k - |M|$;
  \\
  \hspace*{10mm} {\bf case 3:} $U = U\cup M$; $\quad$ $X(M) = \{s:
  c_1, c_2: l, h; t, r\}$; $\quad$ update $A$;

  {\bf PHASE~II.}  $\quad\setminus\!\!\setminus$ now both $G[U]$ and
  $G - U$ are interval graphs.
  \\
  build interval representations for $G - U$;
  \\
  {\bf for each} connected component $M$ in $U$ {\bf do}
  \\
  a \hspace*{3mm} $p = \mathtt{right}(h(M))$ and $q =
  \mathtt{left}(t(M))$;
  \\
  b \hspace*{3mm} pick $\ell\in [p,q]$ such that no edge in
  $K_{\ell}\times M$ is forbidden and $K_{\ell}$ has the minimum size;
  \\
  c \hspace*{3mm} connect $M$ and $K_{\ell}\setminus N(M)$; $\quad$
  decrease $k$ accordingly; $\quad$ remove $M$ from $U$;
  \\
  {\bf if} $k< 0$ {\bf then return ``NO;''} {\bf else return} $G$.

\end{quote} \vspace*{-6mm} \strut} $$\boxit{\box4}$$
\vspace*{-9mm}
\caption{Outline of algorithm for \textsc{interval completion}}
\label{fig:alg-completion}
\end{figure}
\vspace{-4mm}

\section{Omitted proofs}
\subsection{Proof of Lemma~\ref{lem:6-enough}}
\label{sec:3.5}
The numbers of edges that are eligible to break ATs witnessed by small
AWs in Figure~\ref{fig:fill-edges} are $12, 8, 9, 10, 6, 7,$ and $8$
respectively.  
Observe that small AWs always reveal symmetry property.
\begin{proof}[Proof of Lemma~\ref{lem:6-enough}]
\begin{figure*}[ht]
  \centering
  \begin{subfigure}[b]{0.18\textwidth}
    \centering
    \begin{tikzpicture}[scale=.17]
    \node [corner,label=right:$t_1$] (s) at (0,6.44) {};
    \node [corner,label=above:$t_2$] (a) at (-7, 0) {};
    \node [vertex,label=below:$v_2$] (a1) at (-4, 0) {};
    \node [special,label=below:$c$] (v) at (0, 0) {};
    \node [vertex,label=below:$v_3$] (b1) at (4, 0) {};
    \node [corner,label=above:$t_3$] (b) at (7, 0) {};
    \node [vertex,label=right:$v_1$] (c) at (0,3.5) {};
    \draw[] (a) -- (a1) -- (v) -- (b1) -- (b);
    \draw[] (v) -- (c) -- (s);

    \draw[fill edge] (s) to (v) (a) to (v);
    \draw[fill edge] (v) to (b) (a1) to (b1);
    \draw[fill edge] (a1) -- (c) -- (b1);
    \end{tikzpicture}
  \end{subfigure}%
  \quad
  \begin{subfigure}[b]{0.18\textwidth}
    \centering
    \begin{tikzpicture}[scale=.18]
    \node [corner,label=right:$t_1$] (s) at (0,3.5) {};
    \node [corner,label=above:$t_2$] (a) at (-7, 0) {};
    \node [vertex,label=above:$v_2$] (a1) at (-4, 0) {};
    \node [special,label=45:$c$] (v) at (0, 0) {};
    \node [vertex,label=above:$v_3$] (b1) at (4, 0) {};
    \node [corner,label=above:$t_3$] (b) at (7, 0) {};
    \node [vertex,label=below:$u$] (c) at (0,-3.8) {};
    \draw[] (a) -- (a1) -- (v) -- (b1) -- (b) -- (c) -- (a);
    \draw[] (b1) -- (c) -- (a1);
    \draw[] (v) -- (c) -- (s);

    \draw[fill edge] (s) to (c);
    \draw[fill edge] (b) to (v) to (a) (a1) to (b1);
    \end{tikzpicture}
  \end{subfigure}%
  \quad
  \begin{subfigure}[b]{0.18\textwidth}
    \centering
    \begin{tikzpicture}[scale=.22]
    \node [corner,label=right:$t_1$] (s) at (0,6) {};
    \node [corner,label=above:$t_2$] (a) at (-6, 0) {};
    \node [vertex,label=below:$v_2$] (a1) at (-2.5, 0) {};
    \node [vertex,label=below:$v_3$] (b1) at (2.5, 0) {};
    \node [corner,label=above:$t_3$] (b) at (6, 0) {};
    \node [vertex,label=below:$v_1$] (c) at (0,3.5) {};
    \draw[] (a) -- (a1) -- (b1) -- (b);
    \draw[] (c) -- (s);
    \draw[] (a1) -- (c) -- (b1);

    \draw[fill edge] (a) to (b1) (a1) to (b);
    \draw[fill edge] (a1) -- (s) -- (b1);
    \draw[fill edge] (a) -- (c) -- (b);
    \end{tikzpicture}
  \end{subfigure}%
  \quad
  \begin{subfigure}[b]{0.18\textwidth}
    \centering
    \begin{tikzpicture}[scale=.22]
    \node [corner,label=right:$t_1$] (s) at (0,6) {};
    \node [corner,label=above:$t_2$] (a) at (-6, 0) {};
    \node [vertex,label=below:$v_1$] (a1) at (0, 0) {};
    \node [corner,label=above:$t_3$] (b) at (6, 0) {};
    \node [vertex,label=135:$v_3$] (c1) at (-3,3) {};
    \node [vertex,label=45:$v_2$] (c2) at (3,3) {};
    \draw[] (a) -- (a1) -- (b) -- (c2) -- (s) -- (c1) -- (a);
    \draw[] (c1) -- (c2) -- (a1) -- (c1);

    \draw[fill edge] (a1) -- (s) (b) -- (c1) (a) -- (c2);
    \draw[white, bend right] (-2,0)  to (b);
    \end{tikzpicture}
  \end{subfigure}%
  \quad
  \begin{subfigure}[b]{0.18\textwidth}
    \centering
    \begin{tikzpicture}[scale=.22]
    \node [corner,label=right:$t_1$] (s) at (0,6) {};
    \node [corner,label=above:$t_2$] (a) at (-6, 0) {};
    \node [vertex,label=below:$v_0$] (a1) at (-2, 0) {};
    \node [vertex,label=below:$v_1$] (b1) at (2, 0) {};
    \node [corner,label=above:$t_3$] (b) at (6, 0) {};
    \node [vertex,label=135:$v_3$] (c1) at (-3,3) {};
    \node [vertex,label=45:$v_2$] (c2) at (3,3) {};
    \draw[] (a) -- (a1) -- (b) -- (c2) -- (s) -- (c1) -- (a);
    \draw[] (c1) -- (c2) -- (a1) -- (c1);

    \draw[white, bend right] (-2,0)  to (b); 
    \draw[fill edge] (b1) -- (s);
    \draw[fill edge] (s) -- (a1) (b) -- (c1) (a) -- (c2);
    \draw (c1) -- (b1) -- (c2);
  \end{tikzpicture}
  \end{subfigure}
  \caption{A minimum interval supergraph must contain a dashed edges.}
  \label{fig:fill-edges-labeled}
\end{figure*}
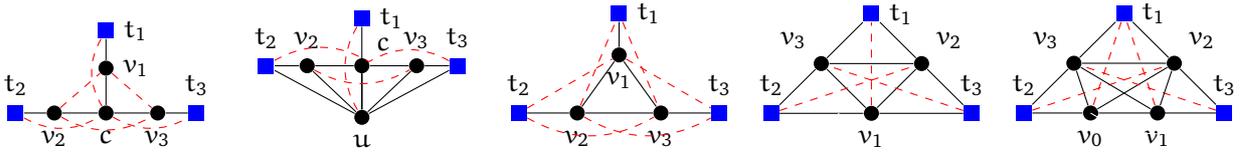

An edge must be inserted between one terminal and the defining path
connecting other terminals.  We redraw the graphs in
Figure~\ref{fig:fill-edges-labeled} and label the vertices for the
easiness of references.

\vspace*{1mm}
\noindent{\bf Long claw.} 
The insertion of edge $t_1 v_2$ will introduce a $4$-hole $(t_1 v_1 c
v_2 t_1)$.  To break this hole we will need at least one of $t_1 c$ and
$v_1 v_2$, which are both included.  Symmetrical arguments apply to
all of $\{t_1 v_3, t_2 v_1, t_2 v_3, t_3 v_1, t_3 v_2\}$.  The
insertion of $t_1 t_2$ will introduce a $5$-hole $(t_1 v_1 c v_2 t_2
t_1)$.  All the five edges required to break this hole is either
included or previously discussed.  Symmetric arguments apply to $t_2
t_3$ and $t_3 t_1$.

\vspace*{1mm}
\noindent{\bf Whipping top.}
The insertion of edge $t_1 t_2$ will introduce a $4$-hole $(t_1 u c
t_2 t_1)$.  To break this hole we will need at least one of $t_1 u$ and
$t_2 c$, which are both included.  A symmetric argument applies to
$t_1 t_3$.  The insertion of edge $t_2 v_3$ will introduce a $4$-hole
$(t_2 v_2 c v_3 t_2)$.  To break this hole we will need at least one of
$t_2 c$ and $v_2 v_3$, which are both included.  A symmetric argument
applies to $t_3 v_2$.  The insertion of $t_2 t_3$ will introduce a
$5$-hole $(t_2 v_2 c v_3 t_3 t_2)$.  To break this hole we will need at
least one of $t_2 c$, $t_3 c$, and $v_2 v_3$, which are all included.

\vspace*{1mm}
\noindent{\bf 2-Net.} The insertion of edge $t_1 t_2$ will introduce a $4$-hole
$(t_1 v_1 v_2 t_2 t_1)$.  To break this hole we will need at least one
of $t_1 v_2$ and $t_2 v_1$, which are both included.  Symmetric
arguments apply to the other two edges $t_2 t_3$ and $t_3 t_1$ that
are not included.

\vspace*{1mm}
\noindent{\bf 1-Tent.} The insertion of edge $t_1 t_2$ will introduce a
$4$-hole $(t_1 t_2 v_1 v_2 t_1)$.  To break this hole we will need at
least one of $t_1 v_1$ and $t_2 v_2$, which are both included.
Symmetric arguments apply to the other two edges $t_2 t_3$ and $t_3
t_1$ that are not included.

\vspace*{1mm}
\noindent{\bf 2-Tent.} The insertion of edge $t_1 t_2$ or $t_1 t_3$
has the same affect for tents.  The insertion of edge $t_2 t_3$ will
introduce a $4$-hole $(t_2 v_0 v_1 t_3 t_2)$.  To break this hole we
will need at least one of $t_2 v_1$ and $t_3 v_0$.  The insertion of
the former makes $\{t_1,t_2,t_3,v_1,v_2,v_3\}$ a tent.  As shown
above, we need at least one of $\{t_1 v_1, t_2 v_2, t_3 v_3\}$, which
are all included.  A symmetric argument applies to $t_3 v_0$.

Arguments for $3$-nets and $3$-tents are word-for-word copy of that
for long AWs in {Lemma}~\ref{lem:fill-long-aw}.
\end{proof}

\subsection{Proof of Lemma~\ref{lem:shallow-1}}
\label{sec:shallow-1}
    \begin{table*}[ht]
      \footnotesize
      \begin{center}
        \begin{tabular}{r|lccc}
          \hline
          & {} & {$q = p + 1$} & {$q = p + 2$} & {$q > p + 2$} \\
          \hline
          \multirow{8}{*}{$\dag$-AW} &
          $p = 0$  &  $4$-hole &  tent & $\ddag$-AW  \\
          &  &  $(x c b_1 l x)^{*}$ &  $\{l, x, s, c, 
          b_2, b_1\}$ & $(s:x,c: l,b_1\dots b_{q-1},b_{q})^{**}$ 
          \\
          & & & & \\
          & $p = 1$  &  whipping top &  net &
          $\dag$-AW  \\
          &  &  $\{l,b_1,x,
          s,c,b_3,b_2\}$$^{***}$ &  $\{l,b_1,s,x,b_3,b_2\}$ &
          $(s:x:l,b_1\dots b_{q-1},b_{q})^{**}$ \\
          & & & & \\
          & $p > 1$  &  long-claw$^{1}$ &  net& $\dag$-AW 
          \\
          &  &  $\{b_{p-2}, b_{p-1},b_p, s, x, 
          b_{p+2}, b_{p+1}\}$ & $\{b_{p-1},b_p,s,x, b_{q}, b_{q-1}\}$ 
          &  $(s:x:b_{p-1},b_p\dots b_{q-1},b_{q})$
          \\
          \hline
          \multirow{8}{*}{$\ddag$-AW} &
          $p = 0$  &  $4$-hole &  tent & $\ddag$-AW  \\
          & &  $(x c_2 b_1 l x)^{*}$ &  $\{l, x, s, c_2, 
          b_2, b_1\}$ & $(s:x,c_2: l,b_1\dots b_{q-1},b_{q})^{**}$ \\
          &&&&\\
          & $p = 1$  &  whipping top &  net &
          $\dag$-AW  \\
          & &  $\{l,b_1,x,s,c_2,b_3,b_2\}^{***}$ &  
          $\{l,b_1,s,x,b_3,b_2\}$ & $(s:x:l,b_1\dots b_{q-1},b_{q})^{**}$ \\
          &&&&\\          
          & $p > 1$  &  long-claw &  net& $\dag$-AW \\
          & & $\{b_{p-2}, b_{p-1},b_p, s, x, 
          b_{p+2}, b_{p+1}\}$ & $\{b_{p-1},b_p,s,x, b_{q}, b_{q-1}\}$ 
          &  $(s:x:b_{p-1},b_p\dots b_{q-1},b_{q})$\\
          \hline
        \end{tabular}
      \end{center}

      \hspace*{20mm}{*} : The vertex $x$ is in  category ``none.''
      \\
      \hspace*{18.5mm}{**} : The vertex $x$ would be in  category ``full'' if 
      $q = d + 1$.
      \\
      \hspace*{17mm}{***} : A $4$-hole $(x b_p b_{p+1} b_{p+2} x)$ 
      would be introduced if $x \sim b_{p+2}$;

      \caption{Structures used in the proof of Lemma~\ref{lem:shallow-1}
        (category ``partial'' )}
      \label{fig:neighbor-of-shallow-terminal}
    \end{table*}
\begin{proof}[Proof of Lemma~\ref{lem:shallow-1}]
  Suppose to the contrary of statement (1), and without loss of
  generality, $x \not\sim c_2$.  If $x \sim b_i$ for some $1 \leq i
  \leq d$ then there is a $4$-hole $(x s c_2 b_i x)$.  Hence we may
  assume $x \not\sim B$.  There is
    \begin{inparaitem}
    \item a $5$-hole $(x s c b_1 l x)$ or $(x s c b_d r x)$ if $W$ is
      a $\dag$-AW, and $x\sim l$ or $x\sim r$, respectively; 
    \item a $5$-hole $(x s c_2 b_1 l x)$ or $4$-hole $(x s c_2 r x)$
      if $W$ is a $\ddag$-AW, and $x\sim l$ or $x\sim r$, respectively; 
    \item a long-claw $\{x,s,c,b_1,l,b_d,r\}$ if $W$ is a $\dag$-AW
      and $x\not\sim l,r$;
    \item a net $\{x,s,l,c_1,r,c_2\}$ if $W$ is a $\ddag$-AW and $x
      \not\sim c_1,l,r$; or
    \item a whipping top $\{r,c_2,s,x,c_1,l,b_1\}$ centered at $c_2$
      if $W$ is a $\ddag$-AW and $x\not\sim l,r$, but $x \sim c_1$.
    \end{inparaitem}
    Neither of these cases is possible, and thus statement (1) is
    proved.

    For statement (2), let us handle category ``none'' first.  Note
    that $x$, nonadjacent to $B$, cannot be a center of $W$.  If
    $x\sim l$, then there is a 4-hole $(x c_2 b_1 l x)$.  A
    symmetrical argument will rule out $x\sim r$.  Now that $x$ is
    adjacent to the center(s) but neither base terminals nor base
    vertices of $W$, then $(x:c_1,c_2:l,B,r)$ makes another AW.

    Assume now that $x$ is in category ``full.''  Suppose the contrary
    and $x \not\sim v$ for some $v\in N(s)\setminus \{x\}$.  We have
    already proved in statement (1) that $v$ and $x$ are adjacent to
    the center(s) of $W$ (different from them).  In particular, if one
    of $v$ and $x$ is a center, then they are adjacent. Therefore, we
    can assume that $v$ and $x$ are not centers.  If $v\sim b_i$ for
    some $1\leq i \leq d$, then there is a $4$-hole $(x s v b_i x)$.
    Otherwise, $v\not\sim B$, and it is in category ``none.''  Let
    $W'$ be the AW obtained by replacing $s$ in $W$ by $v$; then by
    Lemma~\ref{lem:common-neighbor-of-base}, $x\not\sim v$ will imply
    the existence of small AW, which is impossible.

    Finally, assume that $x$ is in category ``partial,'' that is, $x
    \sim B$, but $x\not\sim b_i$ for some $1\le i \le d$. In this
    case, we construct the claimed AW as follows.  As the case $x
    \not\sim l$ but $x \sim r$ is symmetric to $x \sim l$ but $x
    \not\sim r$; on the other hand, $x$ is adjacent to both $l$ and
    $r$ will put it to category ``full.''  Hence in the following we
    may assume that $x \not\sim r$.  Let $p$ be the smallest index
    such that $x \sim b_p$, and $q$ be the smallest index such that $p
    < q \leq d+1$ and $x \not\sim b_q$ ($q$ exists by assumptions).
    See Table~\ref{fig:neighbor-of-shallow-terminal} for the
    structures for $\dag$-AW and $\ddag$-AW respectively.  As the
    graph is reduced and contains no small {forbidden induced
      subgraph}, it is immediate from
    Table~\ref{fig:neighbor-of-shallow-terminal} that the case $q > p
    + 2$ holds; otherwise there always exists a small {forbidden
      induced subgraph}.  This completes the categorization of
    vertices in $N(s) \setminus T$.
\end{proof}

\subsection{Proof of Lemma~\ref{lem:shallow-2}}
\label{sec:shallow-2}
In this proof we will use $W = (s:c_1,c_2:l,B,r)$ to denote an AW.
\begin{proof}[Proof of Lemma~\ref{lem:shallow-2}]
  Let $x$ and $y$ be any pair of vertices such that $x \in C$ and $y
  \in M$.  Since $G[M]$ is connected by definition, we can find a
  shortest path $P = (v_0 \dots v_p)$ between $v_0=s$ and $v_p=y$ in
  $G[M]$.  We claim that $P \not\sim B$.  Suppose the contrary and let
  $q$ be the smallest index satisfying $v_q\sim B$; note that $q\ge
  1$. This means that every $v_i$ with $i<q$ is in category ``none''
  of Lemma~\ref{lem:shallow-1}(2). Therefore, applying
  Lemma~\ref{lem:shallow-1}(1,2) on $v_{i}$ and AW
  $(v_{i-1}:c_1,c_2:l,B,r)$ inductively for $i = 1,\dots,q-1$, we
  conclude that there is an AW $W_i = (v_i:c_1,c_2:l,B,r)$ for each $i
  < q$.  One more application of Lemma~\ref{lem:shallow-1}(1) shows
  that $v_q$ is adjacent to the center(s) of $W_{q-1}$ as well.  If
  $v_q$ is adjacent to all vertices of $B$, i.e., in the category
  ``full'' with respect to every $W_i$, then
  Lemma~\ref{lem:shallow-1}(2) on $v_q$ and $W_{q-1}$ implies that
  $v_q$ is adjacent to $v_{q-2}\in N(v_{q-1})$, contradicting the
  assumption that $P$ is shortest. Otherwise (the category
  ``partial''), according to Lemma~\ref{lem:shallow-1}(2), there is
  another AW $W'= (v_{q-1}:c'_1,c'_2:l',B',r')$, where $B' \subset B$,
  and $v_q \in \{c'_1,c'_2\}$.  Now an application of
  Lemma~\ref{lem:shallow-1}(1) on $v_q$ and $W'$ shows that $v_q$ is
  adjacent to $v_{q-2}\in N(v_{q-1})$, again a contradiction.  From
  these contradictions we can conclude $P\not\sim B$.  Applying
  Lemma~\ref{lem:shallow-1} inductively on $v_{i+1}$ and
  $W_i=(v_i:c_1,c_2:l,B,r)$, we get an AW with the same centers for
  every $0 \le i \le p$.

    As $x$ is adjacent to both $s$ and $B$, it cannot be in category
    ``none'' with respect to $W$.  We now separate the discussion
    based on whether $x$ is in the category ``full'' or ``partial.''
    Suppose first that $x$ is in the category ``full''; as $x\in
    N(s)$, Lemma~\ref{lem:shallow-1}(1) implies that $x\sim
    c_1,c_2$. Then applying Lemma~\ref{lem:shallow-1}(2) inductively,
    where $i = 1,\dots,p$, on vertex $x$ and $W_{i-1}$ we get that $x
    \sim v_i$ for every $i \leq p$; in particular, $x \sim v_p$ ($=
    y$).  Suppose now that $x$ is in in category ``partial.'' Then by
    Lemma~\ref{lem:shallow-1}(2), there is an AW $W'_0 =
    (v_0:c'_1,c'_2:l',B',r')$, where $B' \subset B$, and $x\in
    \{c'_1,c'_2\}$. As $P\not\sim B$, we have that $v_i\not\sim B'$
    for any $0\le i\le p$, i.e., $v_i$ is in category ``none'' with
    respect to $W'_0$. Therefore, by an inductive application of
    Lemma~\ref{lem:shallow-1}(2) on the vertex $v_i$ and AW
    $W'_{i-1}=(v_{i-1}:c'_1,c'_2:l',B',r')$ for $i = 1,\dots,p$, we
    conclude that there is an AW $W'_p = (v_p:c'_1,c'_2:l',B',r')$,
    from which $x\sim y$ follows immediately.

    Now we show the second assertion.  For any pair of vertices $x$
    and $y$ in $C$, we apply Lemma~\ref{lem:shallow-1} on $x$ and $W$;
    by definition, $x\sim B$ and thus cannot be in category ``none.''
    If $x$ is in category ``full'' with respect to $W$, then
    Lemma~\ref{lem:shallow-1}(2) implies that $x$ is adjacent to $y\in
    N(s)$.  Otherwise, if $x$ is in category ``partial'' with respect
    to $W$, then Lemma~\ref{lem:shallow-1}(2) implies that there is an
    AW $W' = (s:c'_1,c'_2:l',B',r')$ where $B' \subset B$ and
    $x\in\{c'_1,c'_2\}$. Therefore, by Lemma~\ref{lem:shallow-1}(1) on
    the vertex $y\in N(s)$ and $W'$, we get that $y\sim c'_1,c'_2$ and
    hence $x\sim y$.
\end{proof}

\section{Omitted figures}

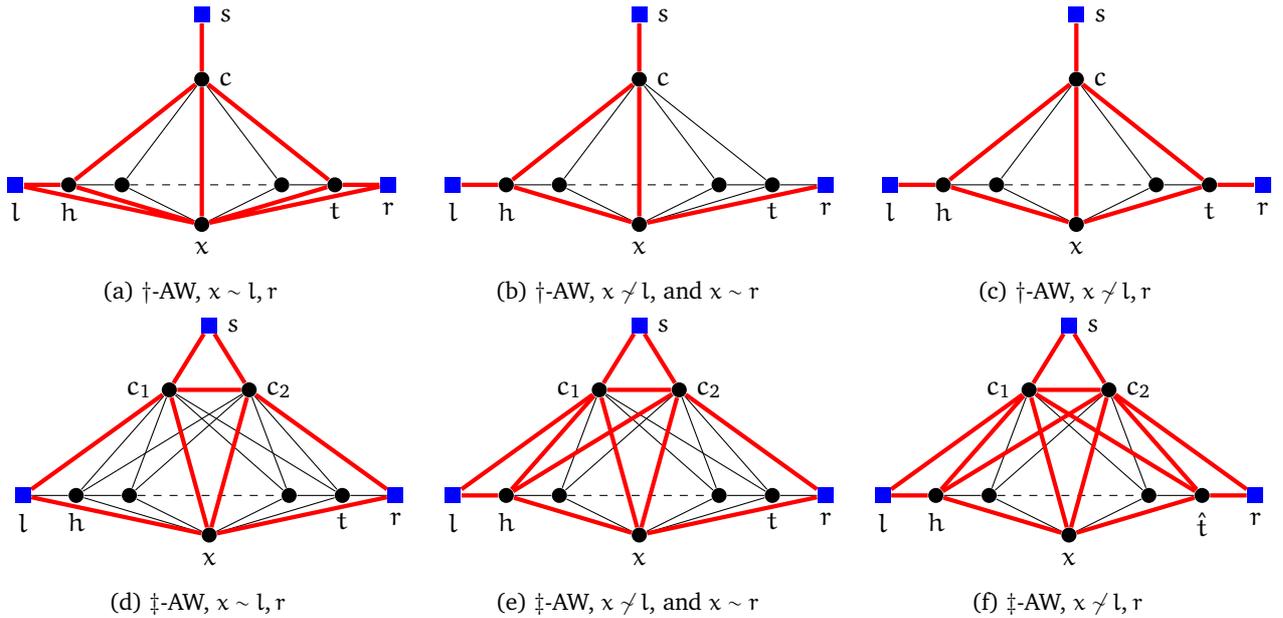
\begin{figure*}[h!]
  \centering
  \begin{subfigure}[b]{0.3\textwidth}
    \centering
    \begin{tikzpicture}[scale=.35]
      \node [corner,label=right:$s$] (s) at (0,6.44) {};
      \node [corner,label=below:$l$] (a) at (-7, 0) {};
      \node [vertex,label=below:$h$] (a1) at (-5, 0) {};
      \node [vertex] (a2) at (-3, 0) {};
      \node [vertex] (b2) at (3, 0) {};
      \node [vertex,label=below:$t$] (b1) at (5, 0) {};
      \node [corner,label=below:$r$] (b) at (7, 0) {};
      \node [special,label=right:$c$] (c) at (0,4) {};
      \node [vertex,label=below:$x$] (u) at (0,-1.5) {};
      \draw[] (a1) -- (a2) -- (c) -- (b2) -- (b1);
      \draw[dashed] (a2) -- (b2);
      
      \draw[] (a2) -- (u) -- (b2);
      \draw[at edge] (a) -- (u) -- (b1) -- (b) -- (u) -- (a1) -- (a);
      \draw[at edge] (a1) -- (c) -- (b1);
      \draw[at edge] (s) -- (c) -- (u);
    \end{tikzpicture}
    \caption{$\dag$-AW, $x\sim l,r$}
  \end{subfigure}
  \qquad
  \begin{subfigure}[b]{0.3\textwidth}
    \centering
    \begin{tikzpicture}[scale=.35]
      \node [corner,label=right:$s$] (s) at (0,6.44) {};
      \node [corner,label=below:$l$] (a) at (-7, 0) {};
      \node [vertex,label=below:$h$] (a1) at (-5, 0) {};
      \node [vertex] (a2) at (-3, 0) {};
      \node [vertex] (b2) at (3, 0) {};
      \node [vertex,label=below:$t$] (b1) at (5, 0) {};
      \node [corner,label=below:$r$] (b) at (7, 0) {};
      \node [vertex,label=right:$c$] (c) at (0,4) {};
      \node [vertex,label=below:$x$] (u) at (0,-1.5) {};
      \draw[] (a1) -- (a2) -- (c) -- (b2) -- (b1) -- (b);
      \draw[dashed] (a2) -- (b2);
      \draw[] (u) -- (b1) -- (c);
      \draw[] (a2) -- (u) -- (b2);

      \draw[at edge] (b) -- (u) -- (a1) -- (a);
      \draw[at edge] (a1) -- (c);
      \draw[at edge] (s) -- (c) -- (u);
    \end{tikzpicture}
    \caption{$\dag$-AW, $x\not\sim l$, and $x\sim r$}
  \end{subfigure}
  \qquad
  \begin{subfigure}[b]{0.3\textwidth}
    \centering
    \begin{tikzpicture}[scale=.35]
      \node [corner,label=right:$s$] (s) at (0,6.44) {};
      \node [corner,label=below:$l$] (a) at (-7, 0) {};
      \node [special,label=below:$h$] (a1) at (-5, 0) {};
      \node [vertex] (a2) at (-3, 0) {};
      \node [vertex] (b2) at (3, 0) {};
      \node [special,label=below:$t$] (b1) at (5, 0) {};
      \node [corner,label=below:$r$] (b) at (7, 0) {};
      \node [special,label=right:$c$] (c) at (0,4) {};
      \node [vertex,label=below:$x$] (u) at (0,-1.5) {};
      \draw[] (a1) -- (a2) -- (c) -- (b2) -- (b1);
      \draw[dashed] (a2) -- (b2);
      
      \draw[] (a2) -- (u) -- (b2);
      \draw[at edge] (u) -- (b1) -- (b);
      \draw[at edge] (a1) -- (c) -- (b1);
      \draw[at edge] (s) -- (c) -- (u) -- (a1) -- (a);
    \end{tikzpicture}
    \caption{$\dag$-AW, $x\not\sim l,r$}
  \end{subfigure}

  \begin{subfigure}[b]{0.3\textwidth}
    \centering
    \begin{tikzpicture}[scale=.35]
      \node [corner,label=right:$s$] (s) at (0,6.44) {};
      \node [corner,label=below:$l$] (a) at (-7, 0) {};
      \node [vertex,label=below:$h$] (a1) at (-5, 0) {};
      \node [vertex] (a2) at (-3, 0) {};
      \node [vertex] (b2) at (3, 0) {};
      \node [vertex,label=below:$t$] (b1) at (5, 0) {};
      \node [corner,label=below:$r$] (b) at (7, 0) {};
      \node [special,label=left:$c_1$] (c1) at (-1.5,4) {};
      \node [special,label=right:$c_2$] (c2) at (1.5,4) {};
      \node [vertex,label=below:$x$] (u) at (0, -1.5) {};
      \draw[] (a) -- (a1) -- (a2) -- (c1) -- (b2) -- (b1) -- (b);
      \draw[] (a1) -- (c1) -- (b1);
      \draw[] (a1) -- (c2) -- (b1);
      \draw[] (a2) -- (c2) -- (b2);
      \draw[] (a1) -- (u) -- (b1);
      \draw[] (a2) -- (u) -- (b2);
      \draw[dashed] (a2) -- (b2);
      \draw[at edge] (a) -- (c1) -- (s) -- (c2) -- (b) -- (u) -- (a);
      \draw[at edge] (u) -- (c1) -- (c2) -- (u);
    \end{tikzpicture}
    \caption{$\ddag$-AW, $x\sim l,r$}
  \end{subfigure}%
  \qquad
  \begin{subfigure}[b]{0.3\textwidth}
    \centering
    \begin{tikzpicture}[scale=.35]
      \node [corner,label=right:$s$] (s) at (0,6.44) {};
      \node [corner,label=below:$l$] (a) at (-7, 0) {};
      \node [vertex,label=below:$h$] (a1) at (-5, 0) {};
      \node [vertex] (a2) at (-3, 0) {};
      \node [vertex] (b2) at (3, 0) {};
      \node [vertex,label=below:$t$] (b1) at (5, 0) {};
      \node [corner,label=below:$r$] (b) at (7, 0) {};
      \node [special,label=left:$c_1$] (c1) at (-1.5,4) {};
      \node [special,label=right:$c_2$] (c2) at (1.5,4) {};
      \node [vertex,label=below:$x$] (u) at (0, -1.5) {};
      \draw[] (a1) -- (a2) -- (c1) -- (b2) -- (b1) -- (b);
      \draw[] (c1) -- (b1) -- (c2);
      \draw[] (a2) -- (c2) -- (b2);
      \draw[] (a1) -- (u) -- (b1);
      \draw[] (a2) -- (u) -- (b2);
      \draw[dashed] (a2) -- (b2);

      \draw[at edge] (a) -- (a1)  -- (u);
      \draw[at edge] (c1) -- (a1) -- (c2);
      \draw[at edge] (a) -- (c1) -- (s) -- (c2) -- (b) -- (u);
      \draw[at edge] (u) -- (c1) -- (c2) -- (u);
    \end{tikzpicture}
    \caption{$\ddag$-AW, $x\not\sim l$, and $x\sim r$}
  \end{subfigure}%
  \qquad
  \begin{subfigure}[b]{0.3\textwidth}
    \centering
    \begin{tikzpicture}[scale=.35]
      \node [corner,label=right:$s$] (s) at (0,6.44) {};
      \node [corner,label=below:$l$] (a) at (-7, 0) {};
      \node [special,label=below:$h$] (a1) at (-5, 0) {};
      \node [vertex] (a2) at (-3, 0) {};
      \node [vertex] (b2) at (3, 0) {};
      \node [special,label=below:$\hat t$] (b1) at (5, 0) {};
      \node [corner,label=below:$r$] (b) at (7, 0) {};
      \node [special,label=left:$c_1$] (c1) at (-1.5,4) {};
      \node [special,label=right:$c_2$] (c2) at (1.5,4) {};
      \node [vertex,label=below:$x$] (u) at (0, -1.5) {};
      \draw[] (a1) -- (a2) -- (c1) -- (b2) -- (b1);
      \draw[] (a2) -- (c2) -- (b2);
      \draw[] (a1) -- (u) -- (b1);
      \draw[] (a2) -- (u) -- (b2);
      \draw[dashed] (a2) -- (b2);

      \draw[at edge] (a1) -- (c1) -- (b1);
      \draw[at edge] (a1) -- (c2) -- (b1);
      \draw[at edge] (a) -- (a1) -- (u) -- (b1) -- (b);
      \draw[at edge] (a) -- (c1) -- (s) -- (c2) -- (b);
      \draw[at edge] (u) -- (c1) -- (c2) -- (u);
    \end{tikzpicture}
    \caption{$\ddag$-AW, $x\not\sim l,r$}
  \end{subfigure}%

  \caption{Adjacency between a common neighbor $x$ of $B$ and
    $s$. [{Lem}.~\ref{lem:common-neighbor-of-base}]}
  \label{fig:common-neighbor}
\end{figure*}

\end{document}

\iftagged{full}{
  \paragraph{Disclaimer.}  In an unpublished manuscript, Arash Rafiey
  claimed a similar result with a totally different approach, which,
  however, contains too many mistakes and logic gaps.  First of all,
  let us consider.  The most crucial bug arises exactly from Lemma
  5.7, which is the most important one and actually the base of the
  whole algorithm.  It states that insertion of some edges $E_X$ will
  not introduce new forbidden subgraphs to $G$, but when it is
  referred in Lemma 5.8, the author assumed that a superset of $E_X$,
  that is, $E_X$ will not introduce forbidden subgraphs to $G +
  (F\setminus E_X)$.
}